\documentclass[sigconf]{acmart}
\copyrightyear{2026}
\acmYear{2026}
\setcopyright{cc}
\setcctype{by}
\acmConference[KDD 2026] {Proceedings of the 32nd ACM SIGKDD Conference on Knowledge Discovery and Data Mining V.1}{August 9--13, 2026}{Jeju Island, Republic of Korea.}
\acmBooktitle{Proceedings of the 32nd ACM SIGKDD Conference on Knowledge Discovery and Data Mining V.1 (KDD 2026), August 9--13, 2026, Jeju Island, Republic of Korea}
\acmISBN{979-8-4007-2258-5/2026/08}
\acmDOI{10.1145/3770854.3783945}

\AtBeginDocument{%
  \providecommand\BibTeX{{%
    \normalfont B\kern-0.5em{\scshape i\kern-0.25em b}\kern-0.8em\TeX}}}
\usepackage[english]{babel}
\usepackage[utf8]{inputenc} % allow utf-8 input
\usepackage[T1]{fontenc}    % use 8-bit T1 fonts
\usepackage{booktabs}       % professional-quality tables
\usepackage{amsfonts}       % blackboard math symbols
\usepackage{nicefrac}       % compact symbols for 1/2, etc.
\usepackage{microtype}      % microtypography
\usepackage{xcolor}         % colors
\usepackage{multicol}

\usepackage{amsmath,wasysym}
\usepackage[ruled,vlined,linesnumbered, resetcount]{algorithm2e}

\usepackage{listings}

\usepackage{subcaption}
\usepackage{multirow}
\usepackage{textcomp}
\usepackage{footnote}
\usepackage{makecell}
\usepackage{xspace}
\usepackage{enumitem}
\usepackage{wrapfig}
\usepackage{amsthm}
\usepackage{wrapfig}


\def\replayalg{{\texttt{IngestReplay}}\xspace}
\def\fetchalg{{\texttt{FetchReplay}}\xspace}

\newcommand{\dataset}[1]{{\mathcal{D}_{#1}}}

\newcommand{\attrib}[1]{\mathcal{A}_{#1}}
\def\numattrib{{M}}
\newcommand{\attribinit}[1]{{Attr_{#1}}}
\newcommand{\attribinitidx}[2]{{Attr_{#1}[#2]}}

\newcommand{\chset}[1]{{\mathcal{C}_{#1}}}
\newcommand{\cohinit}[1]{{C({#1})}}

\newcommand{\alg}[1]{Alg(#1)}

\def\nummetric{{K}}
\newcommand{\metricinit}[1]{{\mathbf{m}_{#1}}}
\newcommand{\metricinitidx}[2]{{m_{#1}[#2]}}

\def\metricset{{\mathrm{M}}}

\newcommand{\replay}[1]{{Repl({#1})}}

\usepackage{xcolor}
\newcommand{\harsha}[1]{{\color{red} (Harsha: #1)}}

\newcommand{\vyas}[1]{{\color{blue} (VS: #1)}}

\newtheorem{theorem}{Theorem}
\newtheorem{definition}{Definition}

\newtheorem{lemma}[theorem]{Lemma}

\newcommand{\model}{\textsc{AHA}\xspace}
\newcommand{\rawingest}{\textsc{StoreRaw}\xspace}
\newcommand{\cubeingest}{\textsc{StoreOutput}\xspace}
\newcommand{\sketchsum}{\textsc{Sketching}\xspace}
\newcommand{\kvstore}{\textsc{KeyValueStore}\xspace}

\newcommand{\unifsample}{\textsc{Sampling}\xspace}

\newcommand{\subgroups}{subgroups\xspace}
\newcommand{\subgroup}{subgroup\xspace}
\newcommand{\conviva}{\textsc{VideoAnalytics}\xspace}

\newcommand{\trace}{\textsc{UmassTrace}\xspace}
\newcommand{\taxi}{\textsc{NYTaxi}\xspace}
\newcommand{\synth}{\textsc{Synth}\xspace}
\newcommand{\chbench}{\textsc{CH-Bench}\xspace}
\newcommand{\mgbench}{\textsc{MGBench}\xspace}
\newcommand{\bigdata}{\textsc{BigData}\xspace}
\newcommand{\imdb}{\textsc{IMDB}\xspace}

\newcommand{\meansigma}{\textsc{3Sigma}\xspace}

\newcommand{\knn}{\textsc{KNN}\xspace}
\newcommand{\isof}{\textsc{IsoForest}\xspace}
\newcommand{\hide}[1]{{}}

\newcounter{findinglabel}
\newcounter{findingnmbr}
\renewcommand{\thefindinglabel}{\textbf{\thefindingnmbr}}

\newcounter{insightnmbr}
\newcommand{\theinsightlabel}{\textbf{\theinsightnmbr}}

\newenvironment{insight}[1]{
    \begin{list}{\textbf{Insight }\theinsightlabel {\em\xspace #1}:~}{\stepcounter{insightnmbr}\setlength{\labelwidth}{0pt}\setlength
        {\labelsep}{0pt}\setlength{\leftmargin}{0in}\noindent\rule{\linewidth}{0.5pt}\vspace{-3pt}\item \bf \em}}{\\*[-7pt]\end{list}\vspace{-10pt}\noindent\rule{\linewidth}{0.5pt}}

\title{\model:  Scalable  Alternative History Analysis for Operational Timeseries Applications}

\author{Harshavardhan Kamarthi}
\affiliation{%
  \institution{Georgia Institute of Technology}
  \city{Atlanta}
  \country{USA}
}
\email{hkamarthi3@gatech.edu}

\author{Harshil Shah}
\affiliation{%
  \institution{Conviva}
  \city{Foster City}
  \country{USA}
}
\email{hshah@conviva.com}

\author{Henry Milner}
\affiliation{%
  \institution{Conviva}
  \city{Foster City}
  \country{USA}
}
\email{hmilner@conviva.com}

\author{Sayan Sinha}
\affiliation{%
  \institution{Georgia Institute of Technology}
  \city{Atlanta}
  \country{USA}
}
\email{sayan.sinha@cc.gatech.edu}

\author{Yan Li}
\affiliation{%
  \institution{Conviva}
  \city{Foster City}
  \country{USA}
}
\email{yan@conviva.com}

\author{B. Aditya Prakash}
\affiliation{%
  \institution{Georgia Institute of Technology}
  \city{Atlanta}
  \country{USA}
}
\email{badityap@cc.gatech.edu}

\author{Vyas Sekar}
\affiliation{%
  \institution{Carnegie Mellon University}
  \city{Pittsburgh}
  \country{USA}
}
\email{vsekar@ece.cmu.edu}


\renewcommand{\shortauthors}{Harshavardhan Kamarthi et al.}

\begin{document}

\begin{abstract}

  Many operational systems collect high-dimensional timeseries data about users/systems on key performance metrics. For instance,  ISPs, content distribution networks, and video delivery services collect quality of experience metrics for user sessions associated with metadata (e.g., location, device, ISP).  Over such historical data,  operators and data analysts often need to run retrospective analysis; e.g., analyze  anomaly detection algorithms, experiment with different configurations for alerts, evaluate new algorithms, and so on. We refer to this class of workloads as {\em alternative history analysis}  for operational datasets. We  show that in such settings,  traditional data processing solutions  (e.g., data warehouses, sampling, sketching, big-data systems) either pose high operational costs or do not guarantee accurate replay. We design and implement a system called AHA (Alternative History Analytics), that overcomes both challenges to provide cost efficiency and fidelity for high-dimensional data.  The design  of AHA is based on analytical and empirical insights about such workloads: 1) the decomposability of underlying statistics; 2)  sparsity in terms of active number of subpopulations over attribute-value combinations;  and 3) efficiency structure of   aggregation  operations in modern analytics databases.  Using multiple real-world datasets and as well as case-studies on production pipelines at a large video analytics company, we show that  AHA provides 100\% accuracy for a broad range of downstream tasks and up to 85$\times$  lower total cost  of ownership (i.e., compute + storage) compared to conventional methods.

\end{abstract}

\maketitle

\section{Introduction}

%% introduce the operational timnesereis problem 
%% what kind of tasks they run typically 

Many operational systems collect high-dimensional data about user and system-level performance indices over time. Analysts need to run diverse algorithms on this collected data; e.g. anomaly detection over subgroups of users grouped by their attributes.  For instance, a large video monitoring system collects  `quality of experience'  metrics  (e.g., bitrate, buffering) for video sessions to find anomalous patterns affecting \subgroups of users~\cite{videoqoe,manousis2022enabling}. These patterns might include
performance degradation affecting
\subgroup; e.g., are users from a ISP-city combination showing degraded performance?

%downstream customers may be interested in {\em alerts} on anomalous groups facing an immediate problem (e.g., quality of experience has dropped in a specific region for a specific user device type) and also {\em longitudinal} analysis for improvement of their overall systems (e.g., identify opportunities for operational efficiencies and improving the quality of experience of users).   

%{\bf TODO: Insert figure for problem and refer in this paragraph}
%  need for alternative history problem   in MLOPS %https://cloud.google.com/architecture/mlops-continuous-delivery-and-automation-pipelines-in-machine-learning
Operational analytics tasks often require the ability to perform  {\em alternative history analytics} over longitudinal  datasets~\cite{googlecloud}.
For instance,   ML scientists may need to do  regression testing (i.e., CI/CD) on historical datasets for benchmarking~\cite{kreuzberger2023machine}.
For anomaly detection, a customer  may query if alerts  triggered several weeks ago would be suppressed by using different sensitivity thresholds~\cite{soldani2022anomaly}.
%Furthermore, data scientists may want to explore different hyperparameter settings for prediction models or evaluate new anomaly detection algorithms.

%% challenges:  cost, unpredictability of downstream task, combinatorial explosion -- use strawmen to explain 
 Supporting alternative history analysis in operational settings is challenging.  
 The datasets and analytics entail significant scale, cost, unpredictability of downstream tasks, and combinatorially large  \subgroups of interest.
Indeed, several seemingly natural solutions fall short. % as seen in Table~\ref{tab:intro:prior}.
%This is in contrast to conventional time-series problems with a single or very small number of time-series to store and analyze.  
As such, achieving {\em low total cost of ownership} and {\em accurate replay for unforeseen tasks} has been elusive.   To see why,  consider the strawman solutions from Table~\ref{tab:intro:prior}.  On one extreme, we can store the raw session measurements (i.e., attributes, metrics) per time step in a database and compute the combinations/statistics of interest when the query is issued.   Unfortunately, there are data retention and cost challenges; e.g.,  there may be terabytes of raw user session data per day. At the other extreme, we can precompute a few attribute combinations (e.g., ``heavy hitters) and statistics of interest. Similarly, another natural approach will be to only store the data about the  alerts triggered. However,  these approaches  lack coverage over future queries that users want to try.

\begin{table}
    \begin{center}
        \begin{footnotesize}
            \begin{tabular}{p{2.5cm}| p{1.5cm}| p{1.5cm} | p{1.5cm}}
                Solution                                    & Low Cost for High Dimensional Data & High fidelity/flexibility & Support for wide range of stats. \\ \hline
                Store raw session  data                                   & No                                 & Yes                       & Yes                              \\ \hline 
            %    Store aggregate statistics of data                                & Yes                                & No                        & No                               \\ \hline 
%                Store only features/predictions                                & Yes                                & No                        & No                               \\ \hline 
                Sample across sessions                                    & Yes                                & No                        & No                               \\ \hline 
                Subpopulation statistics using sketches (e.g. \cite{manousis2022enabling}) & Yes                                & No                        & No                               \\ \hline 
                Key-value store for all subgroup statistics   (e.g.,~\cite{ben2016heavy})      & No                                 & Yes                       & Yes                              \\\hline
                Alternative history Analysis (\model)                                      & Yes                                & Yes                       & Yes                              \\
            \end{tabular}
        \end{footnotesize}
        \caption{Conventional solutions fall short of our key requirements for supporting alternative history analysis for operational time series tasks.}\label{tab:intro:prior}
    \end{center}
\end{table}

\begin{figure}[t]
    \centering
    \includegraphics[width=.85\linewidth]{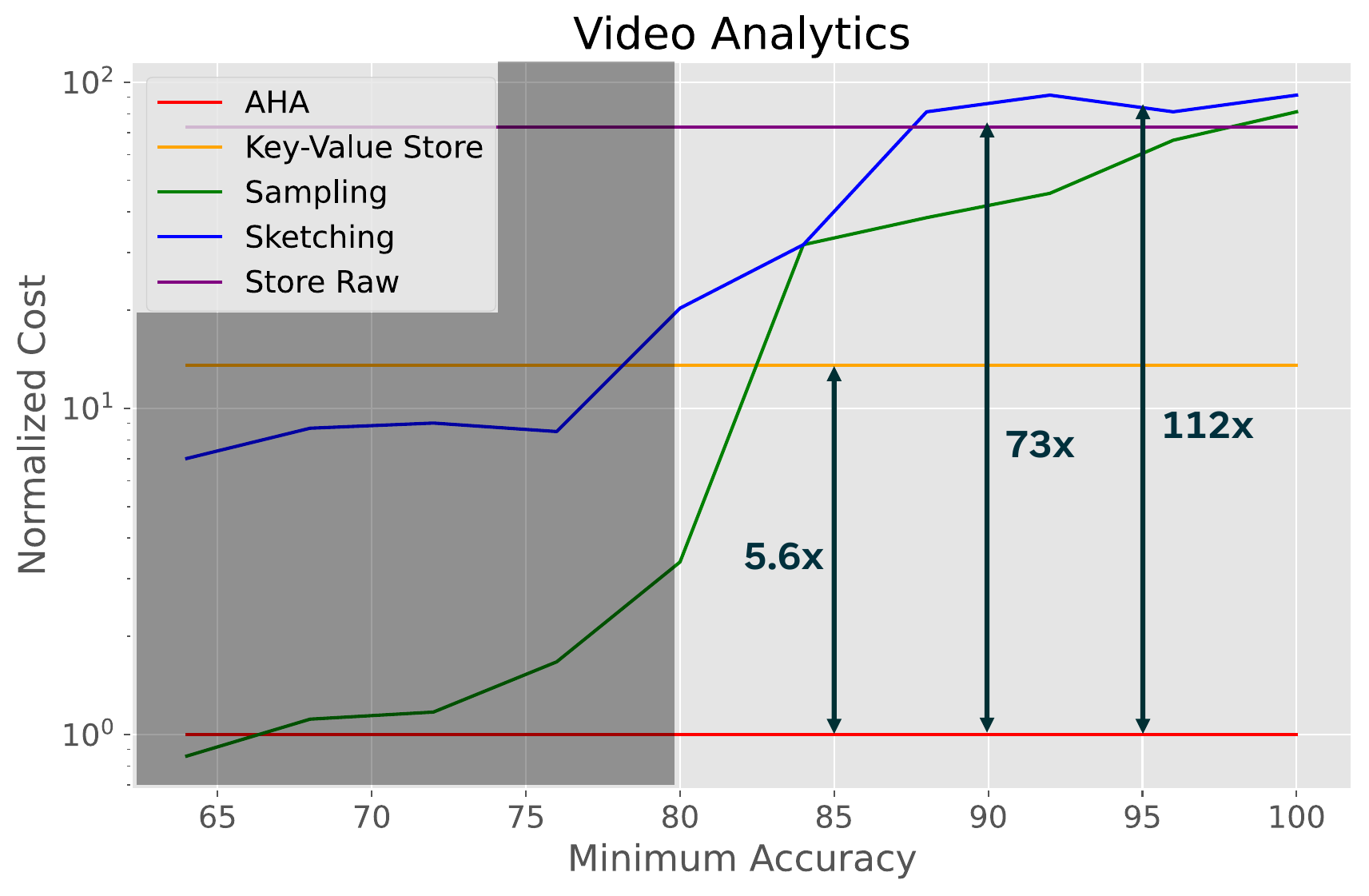}
    \caption{Relative costs of baselines w.r.t \model to reach given minimum accuracy for at least 90 percentile of cohorts. \model is 5.6x better than the best 100\% accurate baseline, 73x cheaper than the default baseline of storing the raw data, and can be over 100x cheaper than approximate solutions to attain a common goal of $> 95\%$ accuracy.}
    \label{fig:mainfig}
\end{figure}
%\harsha{More runs to be monotonically inc., Move legend to not hide AHA, Change color purple-> blue, Y-axis zero, more numbers, log scale in label, quote benefit > 90\%, Shade undesirable range< 85\%, 2 lines 90, 95\% (show x cost), explain percentile of cohorts}
%% system insights 
In this paper, we present the design and implementation of   {\em \model}, a practical system for supporting alternative history analysis in operational systems. The design of \model is based on key analytical and empirical insights on the structure of operational data, query patterns,  the downstream tasks that are typically used in these settings, and  the capabilities of modern data processing systems.

   % \noindent$\bullet$ First, many operational analytics and anomaly detection tasks exhibit a {\em decomposable}  property,  wherein the statistic for a ``parent'' user group can be derived from the statistics of ``children'' groups that constitute the parent group. For example, the mean statistic of a population can
    %be derived from the sum and counts of each of the subpopulations.

    %\noindent$\bullet$ Second,  in any given timestep,  the number of active subgroups that appear is sparse due to the heavy-tailed nature of the distribution of subgroups. Hence, the number of active sessions can often fit in a single node in memory. For example, in case of video analytics, there are only a few highly frequent ISPs that are typically observed in a unit of time compared to the total number of ISPs.

    %\noindent$\bullet$Finally,  when the data can fit inside the resident memory of a single machine, using {\tt CUBE} operations in analytical databases can dramatically outperform     {\tt GROUPBY} operations per subgroup. CUBE operations perform aggregations on all possible combinations of set of rows to provide aggregated metrics across all possible subpopulations. Reproducing this via the standard GROUPBYs may involve redundant operations that are avoided by efficient CUBE implementations in most OLAP platforms.

% \vyas{leaf is undefined. i think this intro needs work to make it more parable.}
%% workflow of what we do based on the key insights 
We create a practical system, where at data ingest time we track necessary statistics of the metrics for multiple (or all possible) subsets of subpopulations, from which metrics for other groups can be derived.
%Their base set of groups are called
%{\em leaf} groups, where each leaf refers to a unique attribute-value combination of all session attributes (e.g., ISP=X, CDN=Y, City= Z).  
Then, when the future alternative analytics history task is issued, we can compute the desired metrics of interest.  
%The statistics required for analysis tasks for all the other groups can  then be  derived from this summary structure over leaf groups that we have retained.   
This decoupled workflow enables us to {\em delay the binding} between the compute-at-ingest and the compute-at-query time to  enable  {\em low total cost of ownership} and {\em accurate replay for unforeseen tasks}.
%key statistics used by a very comprehensive set of analysis methods and tasks.
%We  provide theoretical results establishing the compute and storage efficiency achieved by the two-step delayed binding workflows and characterize the regimes under which \model achieves perfect accuracy.  
%We also empirically evaluate \model against a number of baselines on six datasets from different application domains including a real-world
%dataset from a large video analytics company
%which processes gigabytes of data per second and tens of thousands of queries per minute.
%% Contributions 
%\vyas{say more about the practical implementation? what system what scale etc }
In summary, our work makes the following  contributions: 
  %  \noindent$\bullet$ 
 (1) \textbf{Formulating alternative history analysis  (\S\ref{sec:motivation}):} To the best of our knowledge,  we are the first to formulate the alternative history analysis in operational time series settings. We identify the key cost and accuracy challenges in supporting this capability at scale;  (2) 
    %\noindent$\bullet$
    \textbf{Design and implementation of \model (\S\ref{sec:model}):} We design and implement \model,
          which enables retrieval of
          wide range of features for any possible group efficiently;  
  %  \noindent$\bullet$ 
  (3) \textbf{Correctness and coverage guarantees (\S\ref{sec:model}):} We provide theoretical guarantees establishing perfect accuracy of \model under a broad spectrum of downstream tasks and methods; and   (4) 
   % \noindent$\bullet$ 
   \textbf{Benchmarking on real-world deployment scenarios (\S\ref{sec:results}):} We benchmark \model and other state-of-art solutions on multiple datasets, including using it in production at a large video analytics company. We show that \model provides 34-85 times less total ownership cost without loss in accuracy compared to baselines, which enables cost savings of over \$0.7M per month. We also provide a deployment-study observing the impact on a production data pipeline with 6.2 times reduced total cost.

\section{Background and Motivation}
\label{sec:motivation}

In this section, we provide background on the structure of timeseries analytics tasks  in operational settings and the  need for alternative history analytics.
We identify  requirements  to
support these use cases and discuss why existing solutions fall short.

%to meet . 

%the case for \model system by first introducing the general setup for multi-dimensional large scale data processing workflows encountered in many real-world application. 

\subsection{Operational Timeseries Analytics}

%\begin{figure*}[h]
%    \centering
%    \includegraphics[width=.9\linewidth]{Figs/FirstFig.pdf}
%    \caption{Multi-group cohort data analytics}
%    \label{fig:cohort}
%\end{figure*}
%\begin{figure}[th]
%    \centering
%    \includegraphics[width=.9\linewidth]{Figs/OperationalAnalytics.pdf}
%    \caption{Multi-group Timeseries analytics requires a multivariate view across multiple groups or subpopulations of users/sessions as defined by user/session attributes.}
%    \label{fig:ops}
%\end{figure}

%\paragraph{Dataset:}

%\paragraph{Background on operational timeseries settings} 
We consider time-series analytics workflows that appear in {\em operational settings}.
At a high level, operational settings deal with {\em high dimensional} data and operators are interested in spatiotemporal insights across multiple {\em user/endpoint \subgroups}  to drive operations.
%for deal with large-scale multi-dimensional multi-group data.
To make this concrete, we consider the real-world example of large video analytics service on which this system is deployed.
Video analytics services perform Quality of Experience (QoE) analysis on user data, enabling content and internet providers
to improve user experience, viewership retention, and satisfaction~\cite{videoqoe}.
 %Analytics is done on over 200 billion video streams per ear recording 5 trillion
 %user events per hour across many customers.
Various video analytics \textit{metrics} are collected from video sessions consisting of many different users with different characteristics.
These metrics include bitrate, number of frames dropped, etc.
Each user's viewing session is also annotated with additional metadata called \emph{attributes}.

Typical analysis workflows involve identifying important patterns in metrics over \textit{groups} of users categorized by their similarities in a subset of their attributes (such as geolocation, ISP, device used, etc.).
For example, when buffering times of users from a state using a specific ISP are unusually high, the ISP can be notified to rectify any network issues. 
Detecting QoE issues and patterns entails detecting anomalies in generated metrics
for each user group determined by the common set of attribute assignments among the users in the group. Typical datasets contain millions of possible user groups based on the number and possible combinations of user attributes. These groups are monitored for patterns in multiple metrics.

%\vyas{start from concrete to go more abstract. id use video analytics as a concrete customer first}
%\paragraph{Motivation}
%\paragraph{s} 
 Similar user and product analytics applications are encountered in other domains such as network analysis, monitoring logs, usage analytics etc. where data from multiple types of users are collected and analyzed in  other domains; e.g., telecommunications, observability, IoT, mobile/ad analytics, and so on~\cite{fernandes2019comprehensive,bergmann2019mvtec,fahim2019anomaly,cook2019anomaly,hasan2019attack,parwez2017big}.

%\vyas{maybe also add product analytics since the community may be more familiar with that}

\subsubsection{Datasets schema and typical query templates}
We define the data schema for subpopulation analytics using video analytics since it is our primary application (See Fig. \ref{fig:aggmetrics}).
The dataset consists of telemetry from \emph{sessions} of different users; e.g.,   video-watching sessions. Other examples could be
user sessions in a mobile application or network flows or TCP sessions.
Each session is annotated with a set of \emph{attributes} that describe the user or the session; e.g.,  user location, ISP, device information, etc.
 Each session is associated with a set of \emph{ per-session KPIs} for each epoch of measurement; e.g., per-minute  QoE metrics like buffering time, bitrate, etc.

Sessions are grouped into \emph{cohorts} based on the attributes of the session.
A cohort is a group of sessions that have the same set of attributes.
For example, a cohort could be sessions from a specific
location or all sessions from a specific ISP.
%There can be multiple cohorts based on the different combinations of attributes.
In our experiments, we observe 300 to over one million cohorts based on the dataset.

Operational tasks are often interested in analyzing the metrics at a cohort-level granularity to see global patterns of good/bad performance.   
Hence, we can aggregate the session measurements, called session metrics, of all sessions in a cohort to get the \emph{cohort summary metrics} (referred to as \emph{metrics}).
Example cohort metrics could be the average buffering time of all sessions in a cohort.

%Our goal is to build a efficient and accurate system to provide predictive analytics on these cohort metrics for these large number of cohorts.

\begin{figure}[h]
    \centering
    \vspace{-.1in}
    \includegraphics[width=0.95\linewidth]{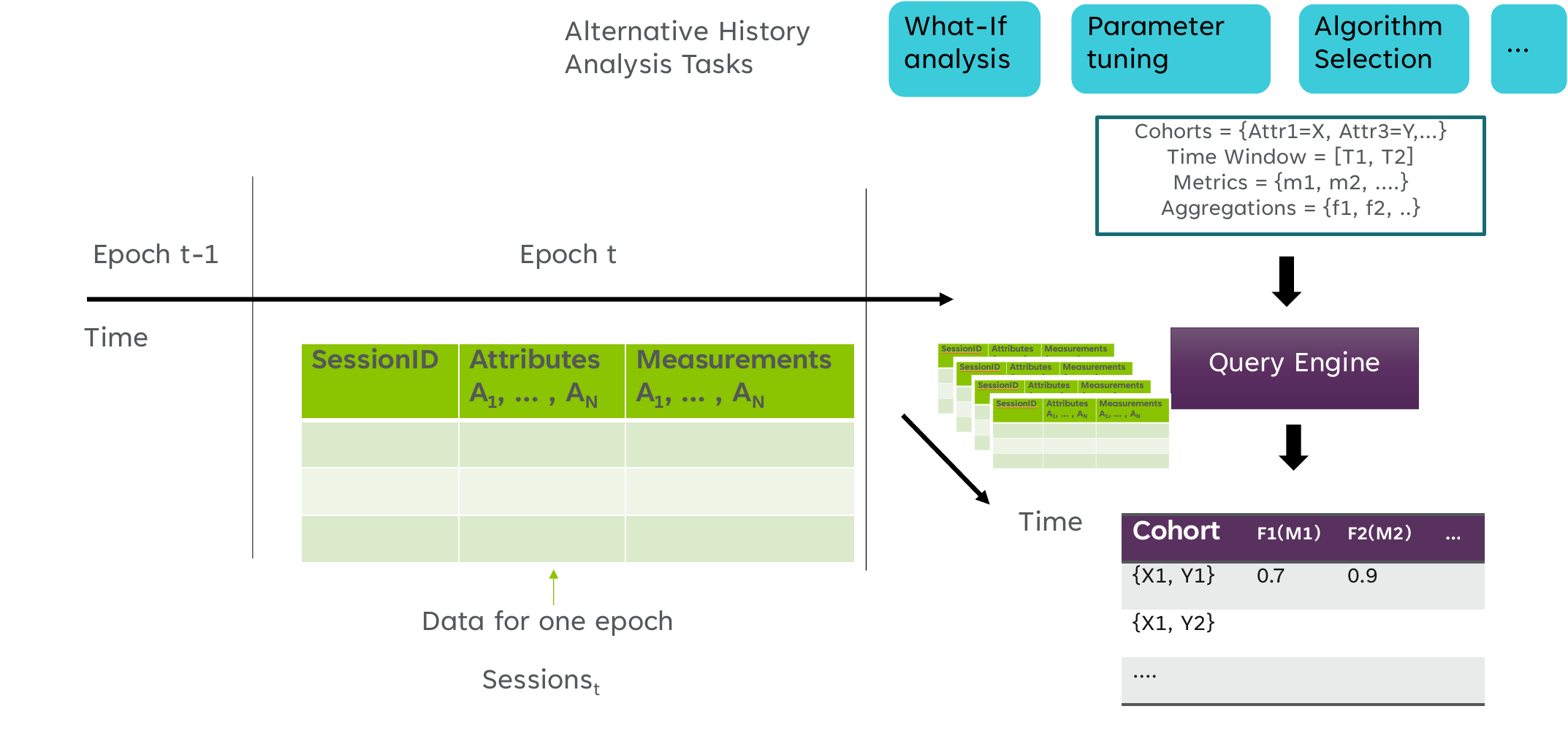}
    \vspace{-.1in}
    \caption{Data setting for Alternative history analysis. users require aggregate statics of arbitrary user cohorts across any time-step in the past.}
    \label{fig:aggmetrics}
    \vspace{-.2in}
\end{figure}

\subsubsection{Need for alternative history analysis}

While such operational data is used for real-time detection of anomalies and forecasting,  our discussions with operators,  algorithm developers, and  MLOps practitioners revealed a more fundamental pain point. 
 In essence, operational systems are constantly in a state of flux due to workload changes, algorithm developments, and customer asks. 
 Consequently, there are numerous use cases that require operators and analysts to access and query longitudinal data as part of everyday operations. 
 Most of the queries can be formulated as a function that derives aggregate statistics and compute predictions
 on them from the session metrics of a cohort.

Formally, let $\mathbf{m}$ be the vector of session metrics at a specific time-step $t$.
and $F(\mathbf{\{m_i\}}_{i\in S})$ be 
  function $F$ that computes aggregate statistics across per-session metrics for the sessions in cohort $S$.
We define cohort $S$ as having attributes $\mathcal{A} = \mathbf{a}$.
For example, when using Spark, we can write PySpark query to compute this statistic and to compute this metric for all cohorts, we can run this query for all possible cohorts:

\begin{lstlisting}
distinct_a = sessions.select("attributes").distinct()
for a_value in cohort_list:
    filtered_df = sessions.filter(sessions.attributes == a_value["attributes"])
    result_df = filtered_df.select(
        F.col("attributes"), 
        F.expr("F(mi_column)")  # Apply function F, replace mi_column with actual column
    )
\end{lstlisting}
%\vyas{this figure should come in our system design? is this AHA system or or aha problem. this feels aha problkem not system design }
Sometimes, we aggregate metrics over time:
\begin{lstlisting}
distinct_a = sessions.select("attributes").distinct()
for a_value in cohort_list:
    filtered_df = sessions.filter(
        (sessions.time > t1) &
        (sessions.time < t2) &
        (sessions.attributes == a_value["attributes"])
    )
    result_df = filtered_df.select(
        F.col("attributes"),
        F.expr("F(mi_column)")  # Replace 'F(mi_column)' with the actual function you are applying
    )
\end{lstlisting}

%\vyas{ figure 4 is misleading  -- why are we focusing on the anomaly detection or alert groups here .. the use cases we have here are more than the alerting right? i think the figure and the use cases dont quite match confused on whether this AHA problem or AHA system and it alos seems too overfitted to alert?}
We describe examples applications  as follows:

\noindent$\bullet$ {\em What-if analysis:} 
Operational time-series analysis needs to provide configuration knobs for users and analysts; e.g., deciding how long an anomaly should persist or how many users are affected.
However, customers need data-driven guidance to understand why predictions are generated, or how changing algorithm parameters would impact the prediction accuracy, false positive rate, etc. based on historical data. For example, the function $F$ could test if a statistic of a measurement like mean login time of an app is greater than a constant threshold. 
Changing the threshold or the statistic to check if certain predictions change would be a common use case.

%settings for alert thresholds 
\noindent$\bullet$ {\em Data-centric regression test in CI/CD for MLOps:}  Given the constant flux in workloads and model drifts, ML engineers and algorithm developers need to refine algorithms and configurations;  e.g., hyperparameter tuning as workloads or baselines change.  When they do so, they need the equivalent of  DevOps best practices such as {\em regression tests}. 
Regression test, for example, could involve testing the algorithm in $F$ for better thresholds and parameters to optimize for accuracy for new datasets based on historical data.

\noindent$\bullet$ {\em Algorithm selection:}  As new time-series methods emerge,   ML teams want to continuously test out novel approaches from the research community.   
However, in practice, there is no one-size-fits-all, and ML teams will need to rigorously test new algorithmic approaches as well as feature engineering in their specific settings (Eg: failure rate of starting a video session) before deploying them.
Testing novel algorithms and statistics required by changing the function $F$ and tuning the model parameters could 
help modellers discover more performant or efficient workflows.
%studies.  
%d to give concrete examples of queries maybe buliding on the previous section on query templates for AHA queries?}

%Such scenarios include understanding why certain anomaly or prediction was generated, trying out a different set of configuration parameters for the algorithm or even trying an entirely new algorithm to improve the quality of analytics systems.
%Since the datasets are ephemeral, we cannot undertake such analysis under these storage restrictions. 

%% describe use case scenarios 

%% describe why this is challenging 

\begin{figure}[th]
    \centering
    \vspace{-.1in}
    \includegraphics[width=.80\linewidth]{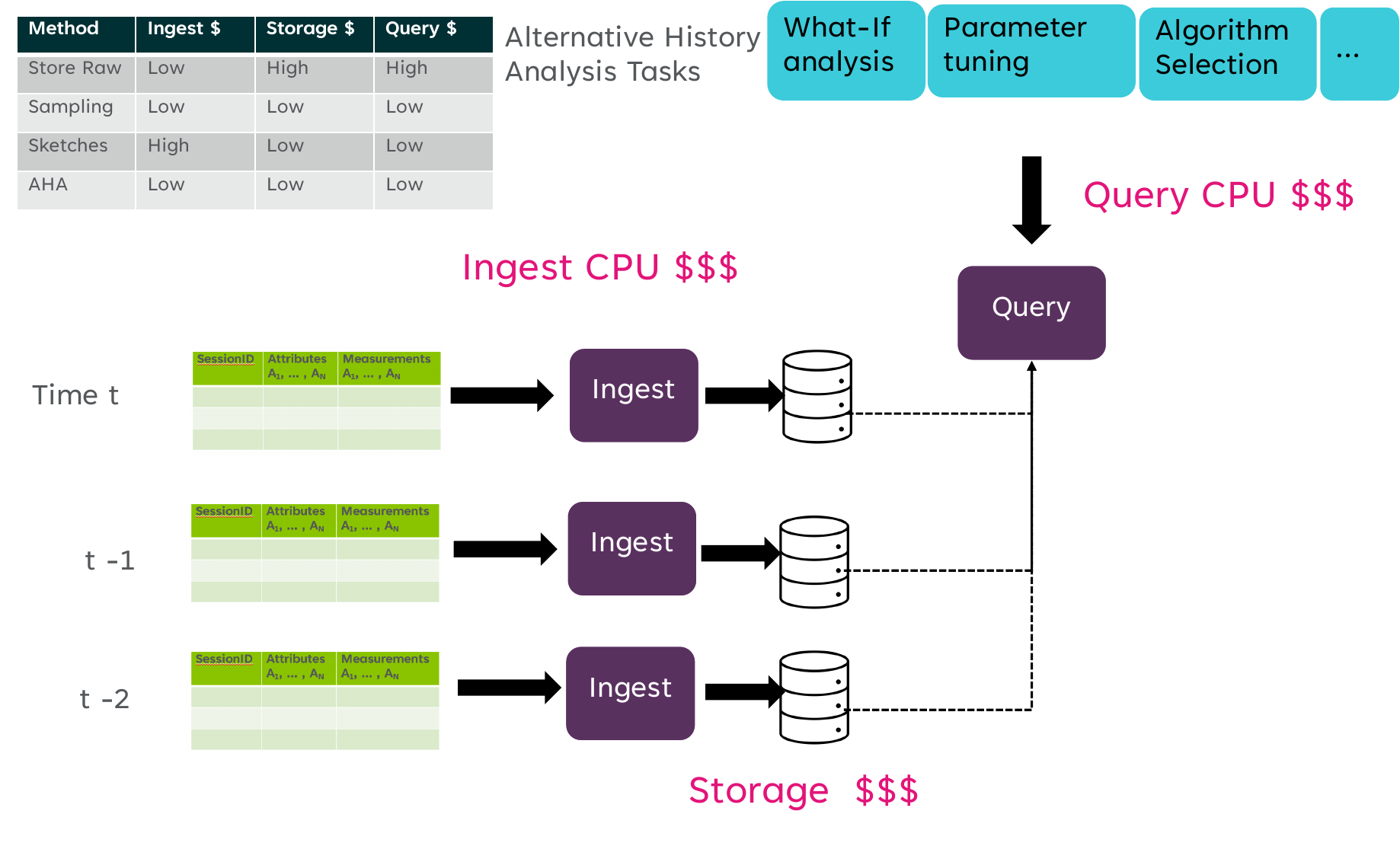}
    \vspace{-.1in}
    \caption{Design space for alternative history analysis: The session data is ingested and stored at each epoch of time. The user queries from the stored data across time for various tasks. \model system computes the required features from the stored summary for the specific application.}
    \label{fig:ReplayNeed}
    \vspace{-.2in}
\end{figure}

\section{Problem Formulation}

%\vyas{it may be useful to break it into two sections here. so sec 2 remains a bit moptivation and here is where we talk abt theoretical definition.}

Given these illustrative scenarios, we formally define the problem of  {\em alternative history analytics}  in  operational timeseries applications.
We begin with some notations and definitions before deriving key requirements. Then, we discuss why canonical solutions from the literature fall short for this class of operational problems. 

\par \noindent \textbf{User groups and Attributes:}
We formally define important aspects of the dataset.
Let the dataset at time $t$ be denoted as $\dataset{t}$.
The dataset consists of data from multiple users (usually in order of millions or tens of millions) at every time-step (minute or second).
Each user is annotated with $\numattrib$ attributes which describe the user-related features. In the video analytics case, these attributes can be user location, ISP,
device information, etc.
Formally, each user $j$ in the dataset is characterized with $\numattrib$ attributes which are initialized
as $\attribinit{j} = \{\attribinitidx{j}{i} \}_{i=1}^\numattrib \in
    \attrib{1} \times \attrib{2} \times \dots \times \attrib{\numattrib}$
where $\attrib{i}$ is the space of all possible values for the attribute $i$.
We make a reasonable assumption that each attribute can only take \textit{discrete values}.
Further, we record $\nummetric$ \textit{metrics} from each user.
The metrics for user $j$ is denoted as $\metricinit{j} = \{\metricinitidx{j}{i}\}_{i=1}^\nummetric$.
Therefore, the dataset at time $t$ is
$\dataset{t} = \{(\attribinit{j}, \metricinit{j})\}_{j=1}^{N_t}$
where $N_t$ is the  number of users tracked at $t$.
Each user \textit{group} is defined by attributes common to all users in the group.
A group $\cohinit{\mathbf{a}}$ is defined by attribute initialization
$\mathbf{a} \in \attrib{1} \cup \{*\} \times \attrib{2} \cup \{*\} \times \dots \times
    \attrib{\numattrib} \cup \{*\}$, i.e,
$\mathbf{a}$ is a sequence of attribute values or $*$ where
 $*$ denotes that the group users can take any
values.
We denote $\mathcal{C}_t$ to be set of non-empty groups from $\dataset{t}$.
However, the number of possible user groups is an exponential of attributes: the number of user groups is the product
of the number of possible subsets of attributes times the number of possible values each of the attributes can take ($\prod_{i=1}^{\numattrib}(|\attrib{i}| +1) - 1$).
Therefore, we use important insights about dataset and analytics requirements
to design \model to be cost-efficient.

%\vyas{i think the problem should be more general than QoE? we should just say session + metrics + metadata to be generic. QoE metrics are specific to video? session granularity can be }

%\vyas{prediction or analytics or estimation?  it seems that somehow we have restricted ourselves to timeseries prediction? but AHA is broader than just alerting algos}

\par \noindent \textbf{Analytics Algorithm:}
For each user group $\cohinit{\mathbf{a}}$,
let $\dataset{t, \mathbf{a}} = \{(Attr_j, \mathbf{m}_j): Attr_j \in \mathbf{a}\}$
be the set of user data for all users in that group.
An algorithm $\alg{\dataset{t, \mathbf{a}}; \theta}$ with configuration parameters
$\theta$
provides an output $\alg{\dataset{t, \mathbf{a}}; \theta} \in \mathbb{R}^b$.
For example, for the task of detecting if a given cohort is anomalous,
$\alg{\dataset{t, \mathbf{a}}; \theta} \in [0,1]$, i.e. it predicts the probability that a given group
is anomalous.
Note that $\dataset{t, \mathbf{a}}$ contains metrics data from a variable number of users, i.e.,
the number of users observed for any group varies across time, with many groups
having very small or zero observed samples for most time-steps.
We can divide the algorithm $\alg{}$ into two parts.
First, we extract the required fixed set of features $F(\dataset{t, \mathbf{a}}; \theta_F)$ from
the dataset $\dataset{t, \mathbf{a}}$.
These features are various statistics collected from the metrics of all the user of the group. They are usually engineered by domain experts leveraging application requirements or domain expertise.
These features are further used by an analytics model $\mathcal{M}(F(\dataset{t, \mathbf{a}}); \theta_M)$ to generate the predictions for the group $\mathbf{a}$.
The model $\mathcal{M}$ of features can be any statistical or machine learning model. Therefore, the algorithm can be described as $\alg{} = <F, \mathcal{M}, \theta>$ where $\theta = (\theta_F, \theta_M)$.
\par \noindent \textbf{Requirements:}
The goal of \model is to support a query $\mathcal{Q}$ of type $<\mathrm{C}, \alg{}, \theta, T>$ which  generate the output for a potentially new algorithm $\alg$ with hyperparameter $\theta$ for specific groups $\mathrm{G}$.
Conceptually, we consider candidate solution consisting of three stages:
(1)  {\em Ingest and summarize} the multi-group data;
(2)  {\em Store} the summarized data to enable a wide range of retrospective analysis;
(3)  {\em Extract} relevant features/information from data summary required by the analysis algorithm $\alg$.

%Any  system consists of an algorithm to ingest session data into summaries, \replayalg, and an algorithm \fetchalg that can fetch relevant information from the summary for any cohort.
%Formally, we design an algorithm \replayalg that summarizes session data %$\dataset{}$
%to $\replay{\dataset{}}$.
%At any future time-step, we retrieve required features $\fetch{\replay{\dataset{}}}$ of all cohorts
%from $\replay{\dataset{}}$ via
%\fetchalg algorithm.

We define the following design objectives for any AH solution:

    \noindent$\bullet$  \textbf{Estimation Equivalence for Future Analytics:}  Given a family of anomaly detection algorithms
          $\alg{} \in \mathrm{Alg}$,
          (\replayalg, \fetchalg) provides task-level equivalence if
          \begin{equation}
              \alg{F(\dataset{});\theta} \approx \alg{F(\replay{\dataset{}});\theta} \forall \theta, \forall \alg{} \in \mathrm{M}
          \end{equation}
          i.e.,
          the estimations of the model if we use session data should be similar
          to estimates using the features from replay. 
          
         \noindent$\bullet$  \textbf{ Equivalence for Cohort metrics:}
          Similarly, (\replayalg, \fetchalg) provides cohort-level equivalence
          if $F(\dataset{}) \approx F(\replay{\dataset{}}) $.
\emph{Note:} Methods that guarantee 100\% accuracy are called \emph{strongly-equivalent}
    solutions. These include deterministic solutions such as key-value stores, methods that accurately calculate the metrics, etc.
    In contrast, sampling and sketching provide approximate estimates. Such methods are called \emph{weakly-equivalent} solutions.
    
    \noindent$\bullet$  \textbf{Low total cost of ownership}:  The total cost of ownership depends on the compute cost for ingest and replay and the storage footprint.
          Ideally, the size of replay data, denoted as $\replay{\dataset{}}$, should be orders of magnitude less than the raw dataset $\dataset{}$:
          $
              |\replay{\dataset{}}| << |\dataset{}|.
          $
          The computational cost of using \replayalg and \fetchalg for deriving the features for models should be similar or significantly better than directly deriving features from raw data $F(\dataset{})$.

%\vyas{ i would suggest we revisit this -- both task-vs-cohort level equivalence, and strong vs weak equivalence? also we need to really tie the tasks back to the to be written query templates in 2.2 }

\par \noindent \textbf{Limitations of related work:}
To tackle this challenge, we  consider  a few seemingly natural straw-man solutions from the literature, and discuss why they fail to satisfy our key requirements .
At a high level, we can consider the design space of options as shown in Figure~\ref{fig:ReplayNeed} in terms of the computation done at {\em ingest} and at {\em query time}. 
Within this framing, we can classify  prior  work that tackles  multidimensional data ingestion and retrieval  into approximate and precise methods. Depending on the type of precise  or approximate analysis performed at ingest/query time, we can have a tradeoff with respect to the ingest, storage, and query time cost. 

\textit{Precise Methods:} Most standard operational pipelines
either compute required features from data stream at ingestion time or later during fetch.
Most of these methods use a database analytics solution like
Spark, Hadoop, etc~\cite{armbrust2015spark,dittrich2012efficient,yang2014druid}.
These SQL and NoSQL databases usually compute the required features precisely and therefore provide 100\% accuracy.
At one end consider the option of doing nothing at ingest, storing the raw session data, and running exact queries at query time. 
Alternatively, rolling up aggregate statistics across multiple subgroups for each possible subpopulation may provide faster compute time, but still requires storing the entire raw data during
ingest to support arbitrary queries.
Both these solutions are expensive with increase in data size to store or with exponential increase in user groups as number of attributes increase as shown in our real-world applications~\cite{agarwal2013blinkdb} in spite of many lossless compression methods used by some of these databases.

\textit{Approximate Methods:} These methods compute approximation of a metric  at ingest or query time usually at lower compute and memory cost . 
Online Sampling is a popular technique that reduces compute and latency costs~\cite{chaiken2008scope,ting2019approximate}. Sketch based methods provide a more efficient bounded tradeoff between memory and accuracy with sublinear complexity.
These frameworks generate approximate summaries during ingest and use
them to estimate statistics with apriori provable error bounds~\cite{manousis2022enabling,cormode2011synopses,durand2003loglog,yang2014druid}.
Analytics databases such as Apache Druid~\cite{yang2014druid} also use these approximate algorithms to compute distinct elements, histograms, etc.
While these methods can provide significant savings in storage cost, they do not guarantee predictive equivalence for sample-poor user groups and scenarios and can take significant computational resources 
in ingesting large high-dimensional datasets.

\hide{
    \paragraph{Relation to time-series anomaly detection methods}

    Typical time-series anomaly detection problems~\cite{darban2022deep} deal with a single univariate or multi-variate time-series.
    In the context of our problem, this problem space can be classified as dealing with a single group $c$ where we only observe a single sample or an estimate of a statistic of $P_{c,t}$ for past time-steps.
    Group Anomaly Detection (GAD)~\cite{NIPS2011_eaae339c} is another class of problems that have properties similar to LAD. However, in GAD the groups do not have lattice partial order relations and the groups are fixed number of groups with independent distributions.
    A class of time-series data that resembles our setting is found in hierarchical time-series  (HTS)~\cite{kamarthi2022profhit} where each time-series is associated with a node of a tree. The non-leaf time-series are derived from the children nodes' time-series via simple aggregation function such as summation. The relation between groups in LAD can be similarly represented by a Hasse diagram over lattice power set which is usually not a tree.
    Further, the samples in LAD are multi-variate where as HTS typically deals with univariate time-series for each node. The characteristics of samples of a cohort is a mixture of samples of its children cohorts in Hasse diagram and therefore many statistics may not follow simple summation constraints observed in HTS.

    \vyas{this seems at odds with what we write in table 1. ideally the prior work we refer to jere is the set of baseliens we are discussing in table 1?  and even more ideally we put our punchline figure showsing us vs them right here or even in intro}
}

\section{\model  Design and Implementation}
\label{sec:model}

In this section, we describe key analytical and empirically validated insights about the nature of the dataset and analytics algorithms in 
most applications. We use these insights to provide a practical basis for building \model to achieve efficiency and accuracy at scale.

\subsection{Main Insights}

%{\bf Insight 1: Big data vs many small data} 

%{\bf Insight 2: Exploit DB features rather than roll out our own}

%{\bf Insight 3: Univ sketches}

%As we saw earlier, traditional approaches to our problem can be broadly divided into two methods:
%1) storing the raw data and computing the statistics on demand, and 2) storing the necessary statistics of all possible subgroups.  To deal with the volume of data, such solutions are implemented using general-purpose platforms such as Spark over distributed file systems or traditional databases.

Alternative history analytics queries are   {\em unpredictable}   over the candidate statistics and subgroups of interest.
Hence, our design philosophy is that of a {\em late binding} architecture from a systems perspective. By delaying the binding of the subgroup summaries and the supported task, we argue \model will be more cost-efficient and scalable.
With this overarching design decision, we focus on the nature of the analytics tasks and datasets to identify useful properties that enable \model  to effectively store and efficiently retrieve relevant data from user subgroups.

\begin{insight}{Task Decomposability\\}
    For many downstream analytics and anomaly detection tasks the required statistics $F(\cdot)$ exhibit {\em decomposable}  property, i.e., the statistic for a ``parent'' user group can be derived from a smaller number of statistics of ``children'' groups that constitute the parent group.
    \label{insight:decomposability}
\end{insight}
%\vyas{instead of saying many can we scope this down to the set of queries that we have already prescribed in sec 2? ie the tasks or functions in queries we saw earlier have this property instead of hedging on few/many etc}

In operational time series settings,  there is a subset relationship across user groups of interest.
For example, a user group corresponding to users with a specific ISP (say Comcast) could be
using many CDNs for a video session.
%The cohort corresponding to users of Comcast corresponds to the union of all groups of users with different CDNs that use Comcast (Fig. \ref{fig:features}).
Moreover, the necessary statistics $F$ required for most analytics algorithms can be derived from statistics of children subgroups.
For example, the mean of a metric for a group of users can be derived from the mean of the metric for each of the children subgroups.
This is a well-known property called decomposability (see \S \ref{sec:decomposability}.
%We formally define this property in Section \ref{sec:decomposability}.
Therefore, in many cases, we can generate statistics for analytics algorithm $F$ for ``parent'' group of US users from statistics collected for groups of each state.
This insight enables us to 1) avoid storage-heavy solutions that need to store data
for all possible subgroups, 2) narrow down the statistics we store via the replay algorithm $\replayalg{}$ for user groups, 3) choose which groups we require to store the statistics for.
\begin{figure}[th!]
    \centering
    \includegraphics[width=.9\linewidth]{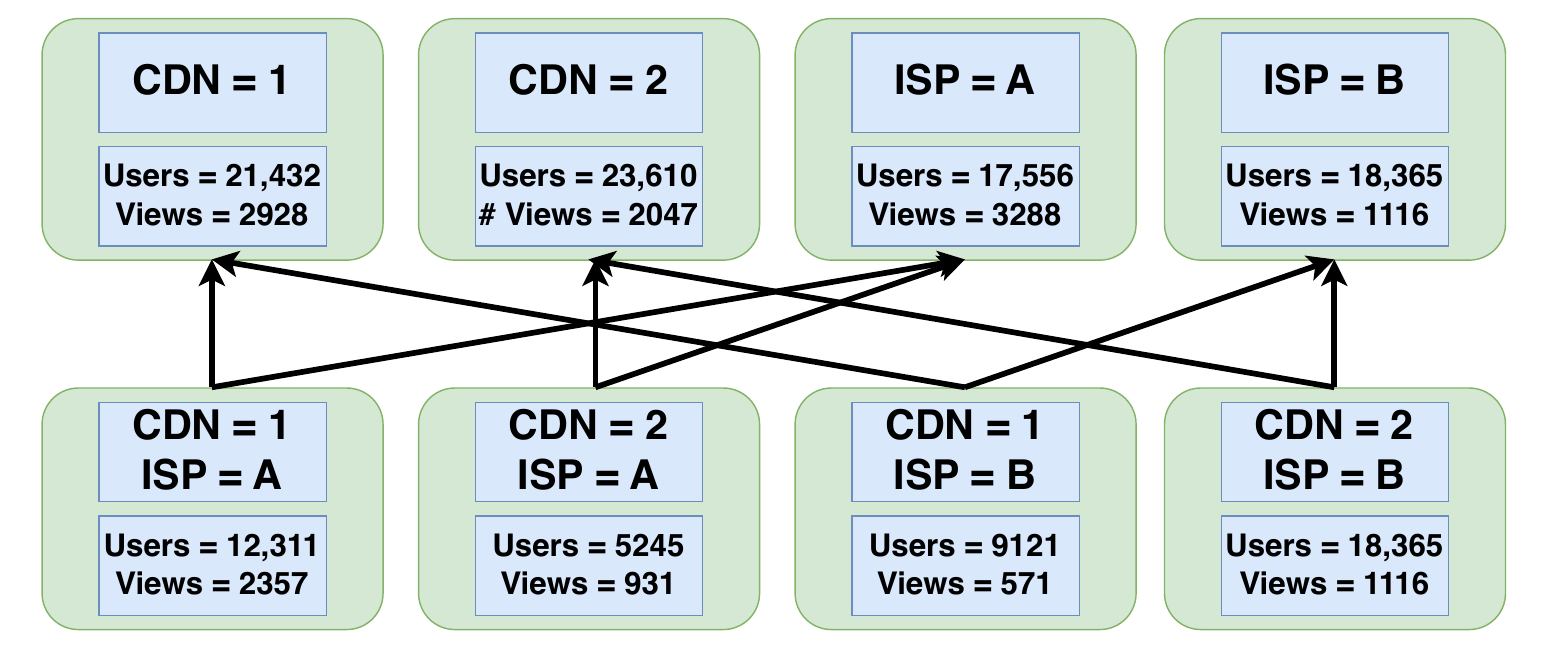}
    \caption{Decomposability: Features of a group can be derived from child groups' features.}
    \label{fig:features}
\end{figure}

%\vyas{the text inside boxes in this fig 4 are not very readable. pls inc fonts?}

%While anomaly detection algorithms and analytics tasks require a wide range of
%statistics across user subgroups, decomposability enables us to compress these demanded statistics into a small list of necessary statistics for some groups.
%The required statistics are generated by simple operations on stored statistics from children subgroups.

Formally, given the session dataset $\dataset{t}$, the final set of inputs to the predictive model is a set of features for \textit{all possible \subgroups} $\mathcal{C}_t$.
The size of $\mathcal{C}_t$ is upper bounded by $O(\prod_{i=1}^{\numattrib} |\mathcal{A}_i|)$ which is an exponential of attributes $\numattrib$.
However, we can generate the necessary statistics $F'$ for the smallest possible set of \subgroups that cover all possible \subgroups called \textbf{LEAF}.
These are \subgroups where all attributes are initialized to a specific value, (i.e., $Leaf(\dataset{}) = \{\cohinit{\mathbf{a}}: \mathbf{a} \in \attrib{1} \times \attrib{2} \times \dots \times \attrib{\numattrib}\}$).
Let $\mathcal{A}_t \subseteq \attribinit{1} \times \attribinit{2}, \dots, \times \attribinit{\numattrib}$
The \model system therefore operates in two stages:

\noindent$\bullet$ \replayalg{} stores the necessary statistics $F'$ of the leaf \subgroups in the replay storage.
\begin{equation}
    \replay{\dataset{t}} = \replayalg{(\dataset{t})} = \bigcup_{\mathbf{a} \in \mathcal{A}_t}F'(\dataset{t, \mathbf{a}}).
\end{equation}

\noindent$\bullet$ \fetchalg{} derives required features $F$ for any user group $C(\mathbf{a})$ from the intermediate features $\bigcup_{\mathbf{a}'\in Child(\mathbf{a})}F'(\dataset{t, \mathbf{a}'})$
where $Child(\mathbf{a})$ are set of leaf groups that are contained in \subgroup $\mathbf{a}$.

\begin{figure}[th!]
    \centering
    \begin{subfigure}[b]{0.48\linewidth}
        \centering
        \includegraphics[width=\textwidth]{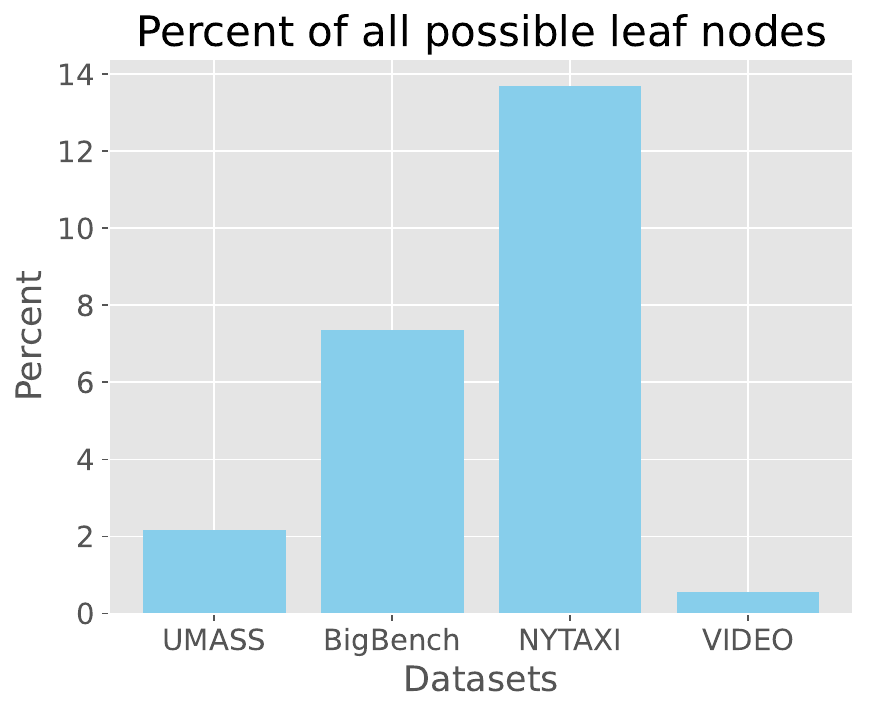}
        \caption{Number of observed leaf groups is smaller than the max. possible leaf groups}
        \label{fig:leafgroupspercent}
    \end{subfigure}
    \hfill % Optional: it adds a space between the figures
    \begin{subfigure}[b]{0.48\linewidth}
        \centering
        \includegraphics[width=\textwidth]{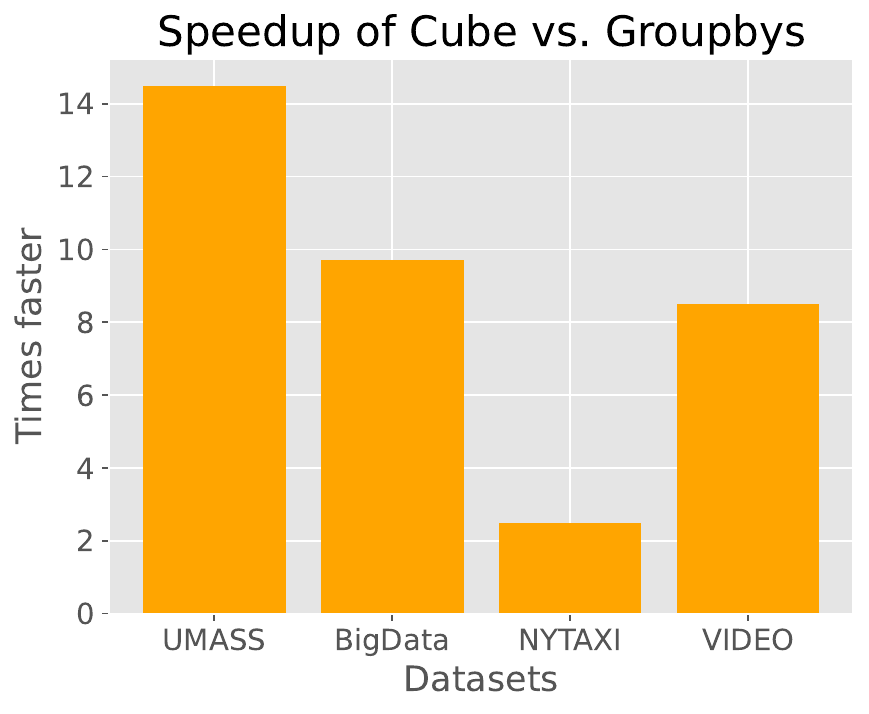}
        \caption{CUBE operation is 3-14 times faster than Groupbys across all datasets}
        \label{fig:cubevsgroupby}
    \end{subfigure}
    \caption{Evidence for subgroup sparsity (Insights \ref{insight:combinationsparsity}) and CUBE efficiency (Insight \ref{insight:cubevsgroupby})}
    \label{fig:scalability_method}
\end{figure}

%\vyas{perhaps explicitly call out an implication for this insight -- based on insight we can do decoupled workflow. and refer to the figure for the design space and the new AHA system figure?}

\begin{insight}{Single node residence due to Active Subgroup Sparsity\\}
    Even though the number of possible subgroups can be exponential in the number of attributes $\numattrib$,
    in any session the number of {\em active subgroups} observed in a given time-step is much smaller.
    \label{insight:combinationsparsity}
\end{insight}

%\vyas{perhaps single node residnce is less critical than active subgroup sparsity. single node resistance may be a corrollary/implication of this insight? maybe explicitly call it oyt as an iumplication or corollary?}

%\vyas{id beef this section up with the theory and empirical result? }

Even with the above insight \ref{insight:decomposability}, however, the total possible number of leaf groups is still potentially exponential in the number of attributes.
In practice, we observe that the number of leaf groups that appear in any given session is significantly smaller, which allows us to fit it into a single node's memory in many cases.
In fact, it is often the case that the number of active user sessions that appear in any given session is likely to fit inside a single node memory.
%This allows us to efficiently perform data processing tasks without communication overheads.
This enables us to decouple the workflow by storing the necessary information from a much smaller number of LEAF subgroups.
This process usually requires a single-node system that can efficiently process the data during \replayalg{} and \fetchalg{} without the overhead of multi-processor communications.
%Then we can derive the required information for other subgroups from LEAF.
Therefore, we only store for LEAF groups that have been observed in the session dataset:
$\mathcal{A}_t = \{\mathbf{a} \in \attribinit{1} \times \attribinit{2}, \dots, \times \attribinit{\numattrib} | |\cohinit{\mathbf{a}}| > 0 \}$.
The average relative number of observed leaf groups to the maximum possible number of leaf groups (Figure \ref{fig:leafgroupspercent})
is 0.5\% to 13.7\%  datasets, which is small enough to fit in a machine with 64GB of memory.

\begin{figure}[htb]
    \centering
    \begin{subfigure}[b]{0.45\linewidth}
        \centering
        \includegraphics[width=\textwidth]{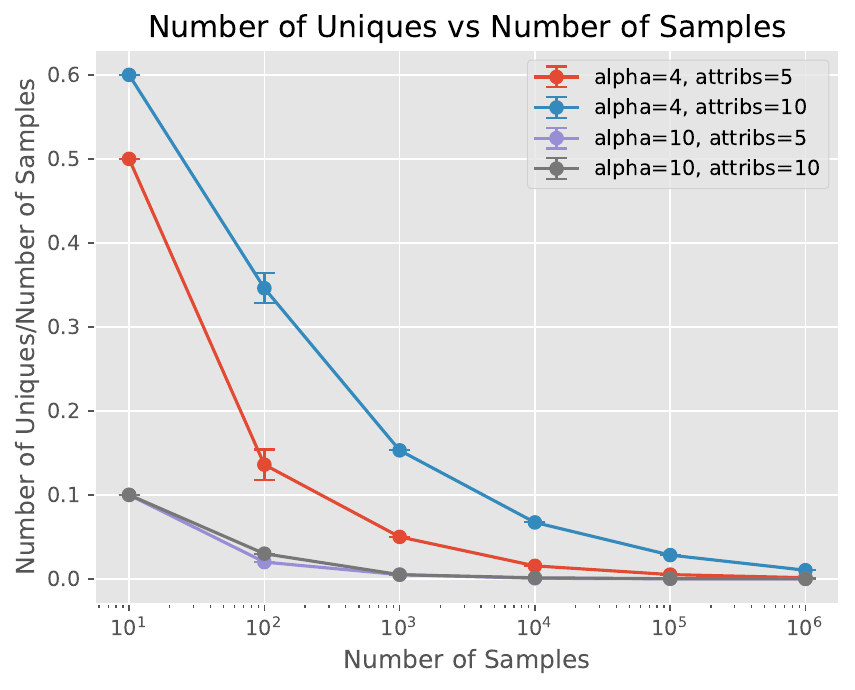}
        \caption{Relative fraction of unique attribute values for different values of $\alpha$ and number of attributes on synthetic data following Zipf distribution}
        \label{fig:synthunique}
    \end{subfigure}
    \hfill
    \begin{subfigure}[b]{0.45\linewidth}
        \centering
        \includegraphics[width=\textwidth]{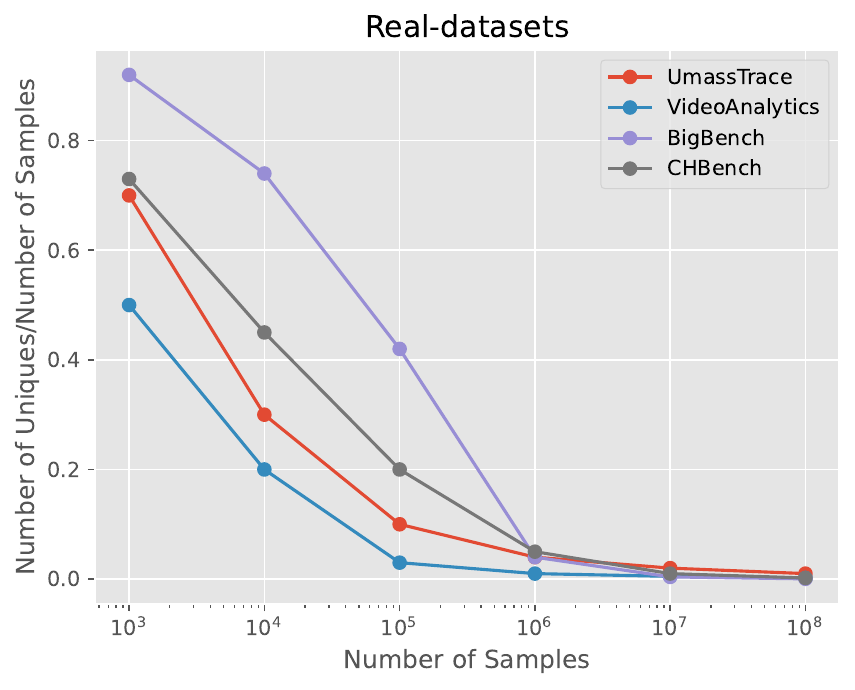}
        \caption{Fraction of unique groups observed in real-world datasets over time across different sample sizes across time}
        \label{fig:realunique}
    \end{subfigure}
    \label{fig:entirefigure}
    \caption{Number of unique sessions (LEAF) increases much slower as sample size increases. We plot the relative fraction of the unique value to the total sample size for synthetic data and real data and see similar patterns of decreasing fractions.}
\end{figure}

%Assume there are $\mathcal{M}$ attributes each with almost $V$ distinct possible values.
%Intuitively, we expect the number of unique attribute values in
%a sample of sessions to be much smaller than the total possible number of possible values that
%can be observed $V^{\mathcal{M}}$
%In many applications, we observe that the distribution of attribute values can be modelled as a power law distribution, especially at the tail, with a small number of attribute values appearing in
%most of the sessions~\cite{zhang2004online,motwani2006distinct}.  

%We also observe that the number of distinct attribute pairs and hence the LEAF groups that are observed in a sample of such sessions tend to decrease with an increase in the number of sessions.
%Empirically, we also observe the fraction of unique groups observed in synthetic datasets following the Zipf distribution in \ref{fig:synthunique} tends to zero.
%Our real-world datasets observe similar trends (Figure \ref{fig:realunique}).

\hide{
Assume there are $\mathcal{M}$ attributes each of which are discrete.
We can represent the attributes as Zipfian distributions with parameter $\alpha$ and $V$ unique values.
The distribution is $P(k) = \frac{1}{Z(\alpha, V)}k^{-\alpha}$ where $Z(\alpha, V)$ is the normalizing constant.
Assume each of the attributes follows the same distribution.
Then, the probability of observing a tuple of attributes $\mathbf{a}$ from a single sampled session in a time-step is  is $P(\mathbf{a}) = \prod_{i=1}^{\mathcal{M}} P(a_i) = \frac{1}{(Z(\alpha, V))^D}\prod_{i=1}^{\mathcal{M}} a_i^{-\alpha}$.
\begin{lemma} \harsha{Assume...}
    The expected number of unique attribute values in $n$ samples is $O(n^{\mathcal{M}/\alpha})$ for large values of $\alpha$.
\end{lemma}
\begin{proof}
    For an attribute value $\mathbf{a}$ let $Y(\mathbf{a})$ be a random variable that is 1 if it is not sampled in the session dataset and 0 otherwise.
    \harsha{Denote random var and derive expectation of S(n)}
    Then, the expectatio for the number of unique attribute values in $n$ samples is $S(n) = \mathbb{E}\left[\sum_{\mathbf{a}}(1-Y(\mathbf{a}))\right]$.
    We note that
        $P(Y(\mathbf{a})) = (1- P(\mathbf{a}))^n$.
    Therefore,
        $S(n) = \sum_{\mathbf{a}}(1 - (1- P(\mathbf{a}))^n)$.
        
We can approximate $(1-P(\mathbf{a}))^n \approx \exp(-nP(\mathbf{a}))$.
Then we can upper bound the expectation as,
\begin{align}
    S(n) &= \sum_{a_1}\sum_{a_2}\dots\sum_{a_{\mathcal{M}}}(1-\exp(-nP((a_1, a_2, \dots, a_{\mathcal{M}})))) \\
    & \leq \int_{a_1=1}^\infty\int_{a_2=1}^\infty\dots\int_{a_{\mathcal{M}}=1}^\infty (1-\exp(-nP((a_1, a_2, \dots, a_{\mathcal{M}})))) da_1 da_2 \dots da_{\mathcal{M}}\\
    & = \int_{a_1=1}^\infty\int_{a_2=1}^\infty\dots\int_{a_{\mathcal{M}}=1}^\infty
    \left(1 - \exp\left( -\frac{n}{Z(\alpha)} \prod_{i=1}^{\mathcal{M}} a_i^{-\alpha} \right) \right) da_1 da_2 \dots da_{\mathcal{M}}
\end{align}

Let $w_1 = \ln a_1$ or $a_1 = e^{w_1}$. Then $da_1 = e^{w_1} dw_1$
\begin{equation}
    S(n) \leq C \int_{w_1=0}^1 \int_{a_2=1}^\infty\dots\int_{a_{\mathcal{M}}=1}^\infty \left(1-\exp\left(-\frac{n e^{-\alpha w_1} \prod_{i=1}^{\mathcal{M}} a_i^{-\alpha}}{Z(\alpha)} \right)\right) e^{w_1} dw_1 da_2 \dots da_{\mathcal{M}}.
\end{equation}
Similarly setting $w_2 = \ln a_2, w_3 = \ln a_3, ...$, we get:
\begin{equation}
    S(n) \leq \int_{w_1=0}^\infty \int_{w_2=0}^\infty\dots \int_{w_{\mathcal{M}}=0}^\infty \left( 1-\exp\left( -\frac{ne^{-\alpha \sum_{i=1}^{\mathcal{M}}w_i}}{Z(\alpha)} \right) \right) 
    e^{\sum_{i=1}^{\mathcal{M}}w_i} dw_1 \dots dw_{\mathcal{M}}
\end{equation}
Let $w = \sum_{i=1}^{\mathcal{M}} w_i$.

\begin{equation}
    S(n) \leq  \int_{w=0}^\infty \left( 1 - \exp\left(-\frac{n}{Z(\alpha)}e^{-\alpha w} \right) \right)\frac{w^{\mathcal{M}-1}}{(\mathcal{M} -1 )!}dw
\end{equation}

Let $t=e^{-\alpha w}$ and we get $dw = -\frac{dt}{\alpha t}$.
%\begin{equation}
%    S(n) \leq \frac{1}{(\mathcal{M} -1 )!}\int_{1}^\infty \left( 1 - \exp\left(-\frac{n}{Z(\alpha)}t \right) \left(-\frac{\ln t}{\alpha} \right)^{\mathcal{M}-1} \right)dt
%\end{equation}

\end{proof}
}

\begin{insight}{CUBE operations are  efficient for data in memory}

    When the dataset can fit inside the resident memory of a single machine,  using  {\tt CUBE} operations in traditional analytical databases can outperform workflows that run  {\tt GROUPBY} operations per subgroup in general-purpose compute engines such as Spark.
    \label{insight:cubevsgroupby}
\end{insight}

A typical solution to compute statistics on subgroups on demand is via {\tt GroupBy} type operations (e.g., atop  Spark or SQL) on a multi-node cluster.
Note, however, in \model, our previous insight \ref{insight:combinationsparsity} means that we can store the necessary statistics for all leaf groups in a single node. In this setting, we can use a more efficient {\tt CUBE}  operator~\cite{gray1997data}. The 'CUBE' operator in SQL and other data processing systems is a special type of GROUP BY capability that generates aggregates for all combinations of values in the selected columns.
It is well studied in database and big data literature~\cite{jiang2016catchtartan,vassiliadis1998modeling}.
Indeed, modern  OLAP engines such as Clickhouse have native support for CUBE operations.
While {\tt CUBE}  operations could be expensive in a distributed setting with a lot of data movements, in a single node memory resident setting it is efficient to derive the required statistics of {\em all possible combinations}  of attribute values of all parent groups from the intermediate statistics of leaf groups (Figure \ref{fig:cubevsgroupby}).
%The {\tt CUBE}  operation computes an aggregation function for all possible combinations of attribute values for
%all possible groups $\mathcal{C}_t$.

%Formally, \model derives required features
%$F()$ for any user group
%$C(\mathbf{a})$ from the intermediate features $\bigcup_{\mathbf{a}'\in Child(\mathbf{a})}F'(\dataset{t, \mathbf{a}'})$
%where $Child(\mathbf{a})$ are set of leaf groups that are contained in \subgroup $\mathbf{a}$.
%We show the average improvement in time taken to compute the required statistics for all possible groups using CUBE operation over groupby operations in .

%\vyas{TODO: need microbenchmark for each insight.}

%\subsection{Design}
\begin{figure}[htb]
    \centering
    \includegraphics[width=0.8\linewidth]{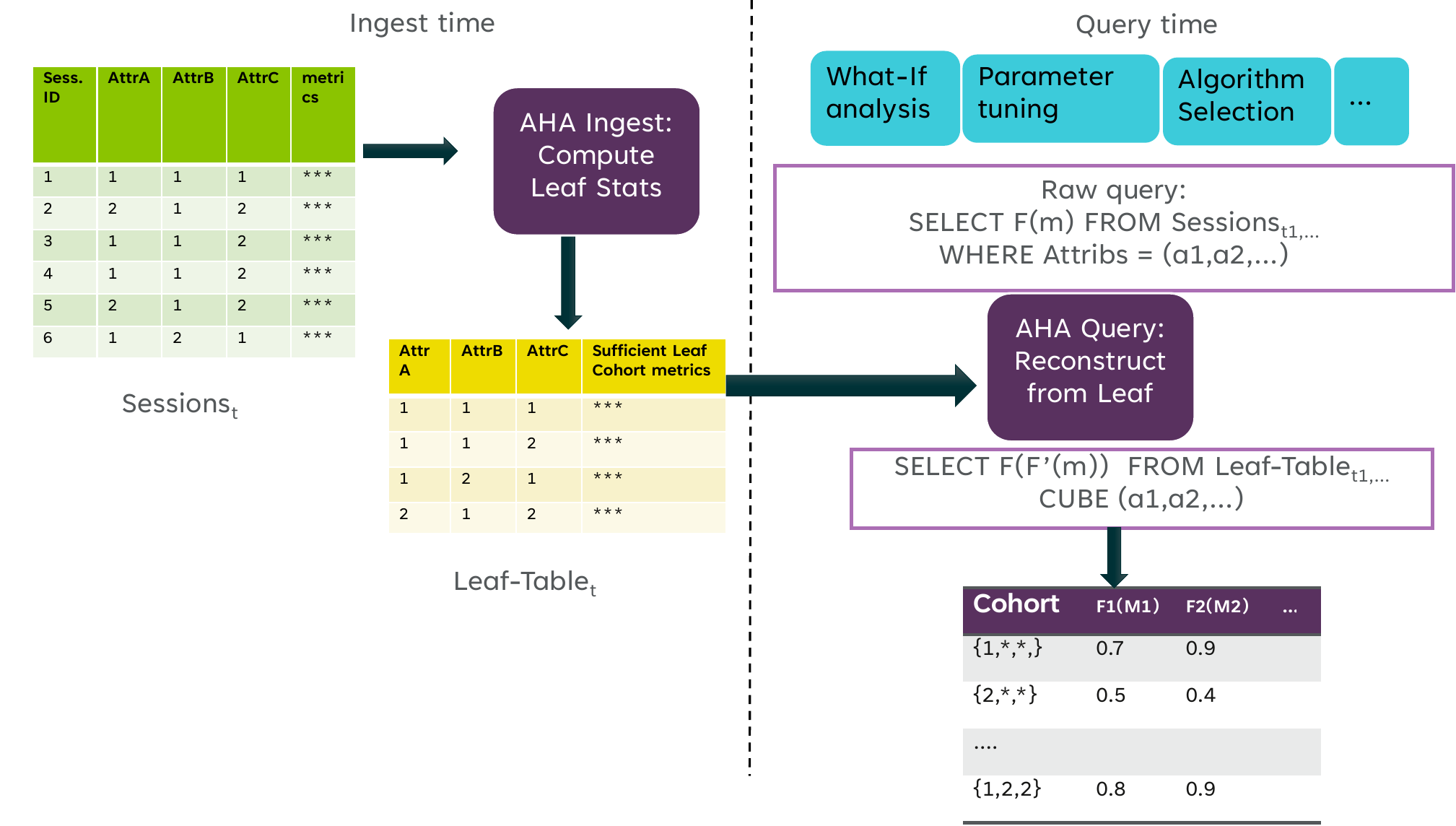}
    \caption{
    %Baseline high fidelity methods compute the required metrics for each cohorts from the session data either to directly store them for future(Store Output) or store the raw data and compute the metrics at inference (Store Raw) causing high storage and compute cost either during ingestion or prediction respectively. 
    \model computes only sufficient metrics for the small number of leaf cohorts during ingest and required metrics later for other cohorts via the efficient CUBE operation.}
    \label{fig:AHAImple}
\end{figure}

\subsection{Design and Implementation  }
%\vyas{uplevel system design and implementation as its section}

Combining the insights, we develop the following workflow for \model consisting of two logical stages (Fig. \ref{fig:AHAImple}):
(1) At ingest, we compute different task-relevant metric statistics of interest per {\em LEAF}.
(2) On query time,  when the alternative history analysis is needed, we run a per-epoch {\em CUBE} over the per-LEAF statistics.
%Our implementation of \model should have specific characteristics to support the design insights.
Assume we need to compute statistics $F$ for predictive algorithm for all cohorts. As specified in The equivalent SQL query is:
\begin{equation}
    \text{SELECT } F(\mathbf{\{m_i\}}_{i\in S}) \text{ FROM sessions WHERE attributes = } \mathbf{a}
\end{equation}
for all possible cohorts $\mathbf{a} \in \in \attrib{1} \cup \{*\} \times \attrib{2} \cup \{*\} \times \dots \times
    \attrib{\numattrib} \cup \{*\}$ which is exponential of attributes $\mathcal{M}$.
Next we will convert this series of queries to a more efficient two-step process proposed in \model
with significantly better storage and compute cost.

\paragraph{Computing LEAF metrics during ingest} We compute the necessary statistics $F'$ of all leaf subgroups from session metric data.
The corresponding query is:
\begin{equation}
\begin{split}
    &\text{CREATE VIEW {\tt leaf-table} as SELECT } F'(\mathbf{\{m_i\}}_{i\in S})
    \\ & \text{ FROM sessions GROUP BY } (\attrib{1}, \attrib{2}, \dots, \attrib{\mathcal{M}})
    \end{split}
    \label{eqn:leafquery}
\end{equation}
The query aggregates all the cohorts in $\{\attrib{1}, \attrib{2}, \dots, \attrib{\mathcal{M}}\}$
{\em observed in the sessions}. As it is much smaller than the possible number of cohorts (Insight \ref{insight:combinationsparsity}), the result is a much
smaller summary dataset.

\paragraph{Generate metrics for any cohort from LEAF metrics during prediction}
We use the metrics of leaf groups in the {\tt leaf-table} to compute sufficient
statistics $F'$ for any  cohorts. For all possible cohorts, we can use the CUBE
operation available in most data processing platforms like Clickhouse and Spark as:
\begin{equation}
    \text{SELECT } F'(F(m)) \text{ FROM {\tt leaf-table} CUBE } (\attrib{1}, \attrib{2}, \dots, \attrib{\mathcal{M}})
\end{equation}
Alternatively, we can also select $F'$ for select combinations of attributes $A'_j \subset \{\attrib{1}, \attrib{2}, \dots, \attrib{\mathcal{M}}\} $ using grouping sets implemented in many databases as:
\begin{equation}
\begin{split}
    \text{SELECT } F'(F(m)) &  \text{ FROM {\tt leaf-table} GROUP BY} \\
    & \text{GROUPING SETS } (A'_1, A'_2, \dots)
\end{split}
\end{equation}
This step is faster than iterating over all possible combination of cohorts due to efficient implementation of CUBE and GROUPING SETS operations in OLAP platforms (Insight \ref{insight:cubevsgroupby}).\footnote{Anonymized code:\url{https://anonymous.4open.science/r/AHA_KDD25-3B63/}}.

%Further, we would like to use a system that can operate on a single node natively (Insight \ref{insight:combinationsparsity}).
%Databases like Spark are designed to operate on multi-node clusters and are not optimized for single-node operations.
%Therefore, we implement   \model using  Clickhouse~\cite{Overview21:online} which is a column-oriented OLAP database that is optimized even for single-node operations and can efficiently perform queries on large datasets.
%We provide code at \url{https://anonymous.4open.science/r/AHA_KDD24-9040}.

%support a wide range of operations for different analytics workflow and is widely used in industry for real-time business analytics.

\subsubsection{Real-time deployment details}
The insights above enabled successful deployment of \model in data engineering pipelines
where large amount of raw data is ingested in real-time from video sessions across major streaming providers. This is used in core data processing pipelines and  alerts system to identify anomalies in time-series across all possible subgroups of users and identify the largest subpopulations and key attributes that are effected by the anomalies.
The system used Clickhouse as the data processing engine and database for its efficiency
on key operations like CUBE.
%\vyas{add some text and figure here for the end to end workflow ?}

\subsection{Soundness of \model}
\label{sec:decomposability}
%\vyas{universality may be too broad a term. maybe talk about scope/genrality/applicability etc?}
%\vyas{i wonder if this should go after deisgn section rather than the insighgts section. we havent really described the design of the system yet}

We characterize the kinds of statistics we can derive using our design.
Let the metric data for any user at any given time be an element of set $\metricset$.
The data for a \subgroup of users is a finite subset of $\metricset$.
A statistic is a function $f: 2^{\metricset} \rightarrow \mathrm{O}$ that maps metric data of a \subgroup
to output domain $\mathrm{O}$. For example, we define mean statistic as $mean = \frac{1}{|\chset{}|}\sum_{m \in \chset{}} m$ which maps to a single scalar.

\begin{definition}[Self-decomposable statistic]
    Let $f$ be a statistic. $f$ is self-decomposable if for any finite set $\mathbf{M}_0 \in 2^{\metricset}$ and any finite disjoint partitioning of $\mathbf{M}_0$ as $\{\mathbf{M}_1, \mathbf{M}_2, \dots, \mathbf{M}_N\}$ (i.e., $\mathbf{M}_i \cap \mathbf{M}_j = \phi$ for any $1\leq i < j \leq N$ and $\bigcup_{i=1}^N \mathbf{M}_i = \mathbf{M}_0$), we have:
    \begin{equation}
        f(\mathbf{M}_0) = f(\{f(\mathbf{M}_1), f(\mathbf{M}_1), \dots, f(\mathbf{M}_N)\}).
    \end{equation}
\end{definition}

For example, \texttt{SUM} is a self-decomposable statistic which adds the metric of given \subgroup. Trivially, the sum of a metric in a \subgroup can be derived as the sum of the sums of the metric in each of the individual mutually disjoint partitions of the \subgroup.

%\vyas{in the sketching world ive seen this called as 'linearity'  \url{https://people.cs.umass.edu/~mcgregor/stocworkshop/mcgregor.pdf }}

\begin{definition}[Decomposable statistic]
    $f$ is a decomposable statistic w.r.t a set of statistics $N(f) = {f_1, f_2, \dots, f_k}$ if there exists
    a function $A_f$
    such that for any finite set $\mathbf{M}_0 \in 2^{\metricset}$ and
    any finite disjoint partitioning of
    $\mathbf{M}_0$ as $\{\mathbf{M}_1, \mathbf{M}_2, \dots, \mathbf{M}_N\}$ we have:
    \begin{equation}
        f(\mathbf{M}_0) = A_f(\{f_1(\mathbf{M}_i)\}_{i=1}^N, \{f_2(\mathbf{M}_i)\}_{i=1}^N,\dots, \{f_k(\mathbf{M}_i)\}_{i=1}^N)
    \end{equation}
\end{definition}
In simple terms, a decomposable statistic is a statistic that can be derived as a function of some statistics of
its partitions.
For example, the statistic \texttt{MEAN} can be derived as the \texttt{SUM}
of values in each partition as well as \texttt{COUNT}, i.e, cardinality of each partition.
%\vyas{is there a proof for these theorems? or is it a proof by construction?}
Decomposable statistics cover a wide range of statistics found in operational pipelines including in our deployment applications. 
\model
support wide range of predictive analytics that require decomposable statistics like mean, sum, range,
variance and other order moments used in commonly used algorithms.
Important non-decomposable metrics can be computed using well-approximated algorithms.
For example, we can compute Histogram sketches that are decomposable are stored for each leaf group and can be
used to compute a wide range of statistics.
Median and other quartiles can be estimated approximately~\cite{greenwald2001space}.
They allow pre-computing intermediate results from different partitions in distributed systems~\cite{yu2009distributed,jesus2014survey,zhang2017symmetric} to reduce the amount of communication across nodes and for faster implementations of these functions in settings such as OLAP databases where directly querying from raw data is expensive.
Moreover, the function $A_f$ is usually a composition of simple aggregation functions such as addition, multiplication, etc. and elementary functions such as square root, logarithm, etc. which are efficiently
implemented in modern databases.

Specifically, we derive the necessary statistics for the smallest possible set of \subgroups that cover all possible \subgroups called \textbf{leaf}. These are \subgroups where all attributes are initialized to a specific value, (i.e., $Leaf(\dataset{}) = \{\cohinit{\mathbf{a}}: \mathbf{a} \in \attrib{1} \times \attrib{2} \times \dots \times \attrib{\numattrib}, |\cohinit{\mathbf{a}}| > 0\}$).
The replay storage consists of $\replay{\dataset{t}} = \{(\mathbf{a}, F'(\dataset{t, \mathbf{a}})): \mathbf{a} \in Leaf(\dataset{t})\}$. here, $F'$ is the set of necessary statistics for the features $F$ of the algorithm: $F' = \bigcup_{f \in F} N(f)$.
This allows us to provide a theoretical guarantee of perfect predictive equivalence with significant real-wprld storage cost reduction.

\begin{theorem}
    If the anomaly detection algorithm $\mathcal{M}$ requires features $F$, all of which are decomposable, then, we can have perfect predictive equivalence by storing only the necessary statistics $F'$ of the leaf \subgroups in the replay storage.
\end{theorem}

%\vyas{this is worth checking with yan and team -- but there is indeed a pilot deployment that switched conviva ai alerts to clickhouse not fully aha but a lot of the ideas here inspired them to move from spark to clickhouse. its useful to ask henry/yan to contribute a para on deployment experience or results?  ie a subset of aha in production is probably a good fillip for the paper in the ADS track? }

\section{Evaluation}
\label{sec:results}

\noindent \textbf{Baselines:}
We compare \model against the following alternative data-processing solutions.
Strong equivalence solutions included are:
  \noindent$\bullet$   \rawingest: A common solution where we store the full raw data during ingest and
        fetch required features during prediction. %We use a standard data processing system in Clickhouse.
%  \item   \cubeingest: Another common solution where we compute the required statistics for all the groups during ingest and do a lookup during fetch.
 %       We store the required statistics in a Spark dataframe.
  \noindent$\bullet$   \kvstore: Similar to \cubeingest, we store only the outputs via a hash map that maps the attributes of groups to stored metrics. We implement the key-value system by extending the hash-map implementation in the Rust programming language.
Solutions with weak equivalence are:.
  \noindent$\bullet$   \unifsample: We sample a small fraction of the user samples for each time-step and fetch required statistics from the samples at fetch.
  \noindent$\bullet$  \sketchsum: We also use a state-of-art sketching solution \textit{Hydra}~\cite{manousis2022enabling} designed for sub-population analytics uses a universal sketch to summarize summary statistics.

\noindent \textbf{Dataset:}
We evaluate \model on multiple datasets from different applications. Statistics on the datasets are discussed in Appendix Table \ref{tab: datasets}.
  \noindent$\bullet$  \conviva: Real-world user video analytics dataset observed at live deployment at a major video-streaming analytics provider .
%  \item \caida: Network flow traces collected from US ISP provider~\cite{TheCAIDA55:online}.
  \noindent$\bullet$ \trace: Network trace dataset from~\cite{NetworkU23:online}.
  \noindent$\bullet$ \taxi: Contains the pickup and dropoff times and other statistics provided by New York Taxi and Limousine Commission~\cite{TLCTripR1:online}.
  \noindent$\bullet$ \chbench: Anonymized DB log traces data provided by Clickhouse
  \noindent$\bullet$ \mgbench: Synthetic benchmark from \cite{ben2017mgbench} originally intended to test GPU performance on data i/o and simple computational workflows
  \noindent$\bullet$ \imdb: Dataset from IMDB containing a list of movies and various features including year of release, actors, language, etc.

\noindent \textbf{Benchmarking details}
We analyze the ability of \model and baselines to support standard prediction pipelines.
We consider anomaly detection algorithms used in production: standard deviation thresholds, $k$ nearest neighbors (\knn)~\cite{amer2012nearest}
and Isolation Forest (\isof)~\cite{liu2008isolation}. All these algorithms use mean of the metrics observed across the users as features.

We measure the fidelity of each of the baselines by measuring both the accuracy of the aggregate metric (metric accuracy) and of the anomaly detection algorithm (task accuracy).
Metric accuracy is measured as RMSE score and task accuracy as classification accuracy of anomalies.
Note that the accuracy here refers to how closely the predictions derived from a given system matches the predictions using raw data.
Methods that compute the metrics precisely including AHA will have 100\% accuracy while approximate methods may have lower accuracy.
To determine the total cost of the system (compute and storage), we measure the cost of performing this operation over a period of a month
on a EC2 system of similar configuration via AWS using S3 to store the
summarized data for each solution. We provide more details on the setup, cost model and parameter settings of the baselines in Appendix \S \ref{sec:benchmark}.

\begin{figure}[htbp]
  \centering
  \begin{subfigure}[b]{0.475\linewidth}
  \centering
    \includegraphics[width=.98\linewidth]{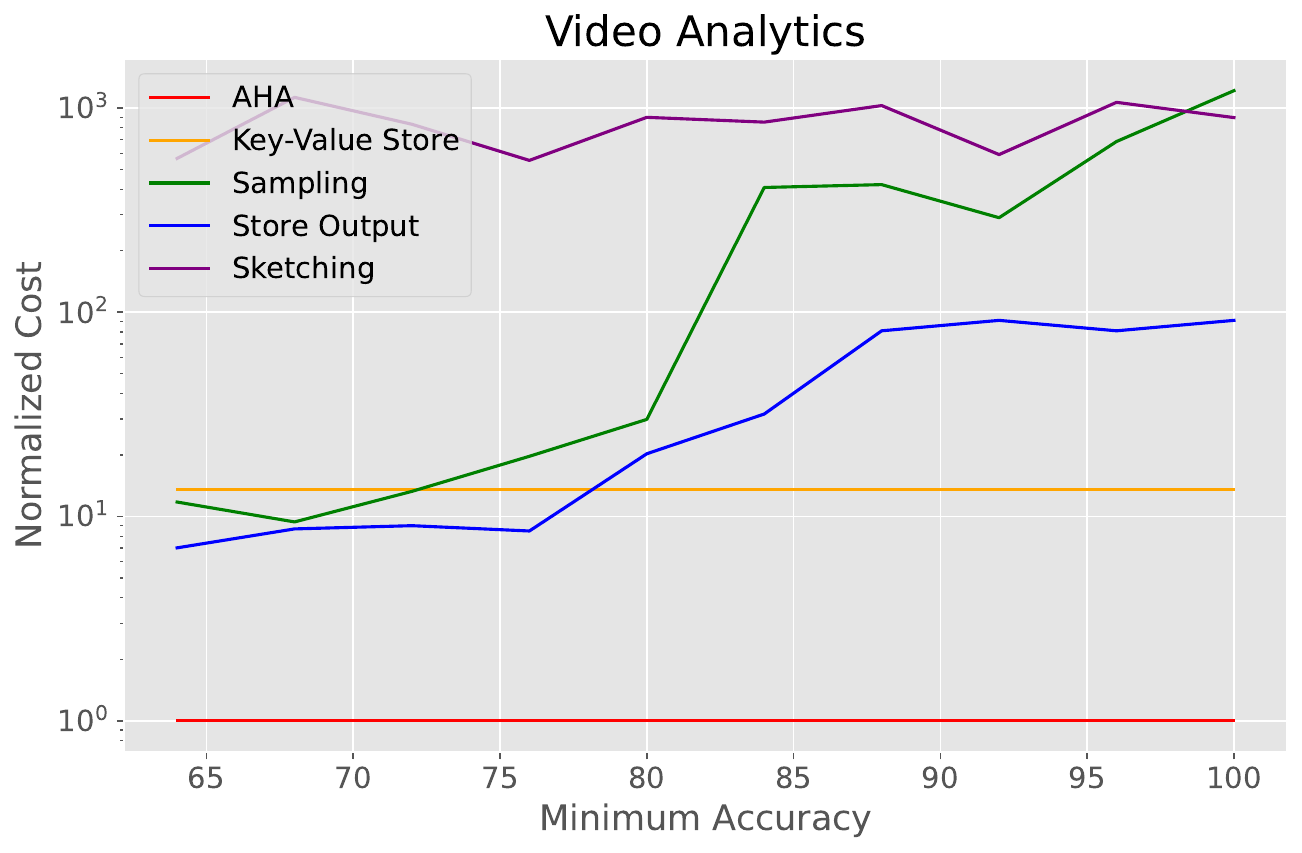}
    \caption{Cost vs Metric accuracy}
  \end{subfigure}%
  \hfill
  \begin{subfigure}[b]{0.475\linewidth}
  \centering
    \includegraphics[width=.98\linewidth]{RoughFigs/q3/videoanalytics.pdf}
    \caption{Cost vs Task accuracy}
  \end{subfigure}%
  \caption{Normalized Cost for different minimum task and metric accuracy requirements. 
  \model is 55-130 times more efficient among strong equivalence solutions and the cheapest solutions for all datasets with over 80\% metric accuracy requirement.  
  \model is lowest cost across all task accuracy requirements (> 90\%) with 8.6-87 times lower cost.
  %Also note that for accuracy > 97\% weak-equivalence methods cost more than even storing full session data showcasing compute overhead of sampling and sketching methods for high accuracy prediction.
  }
  \label{fig:subfigures_task}
\end{figure}

\begin{figure}[h]
  \centering
  \includegraphics[width=.69\linewidth]{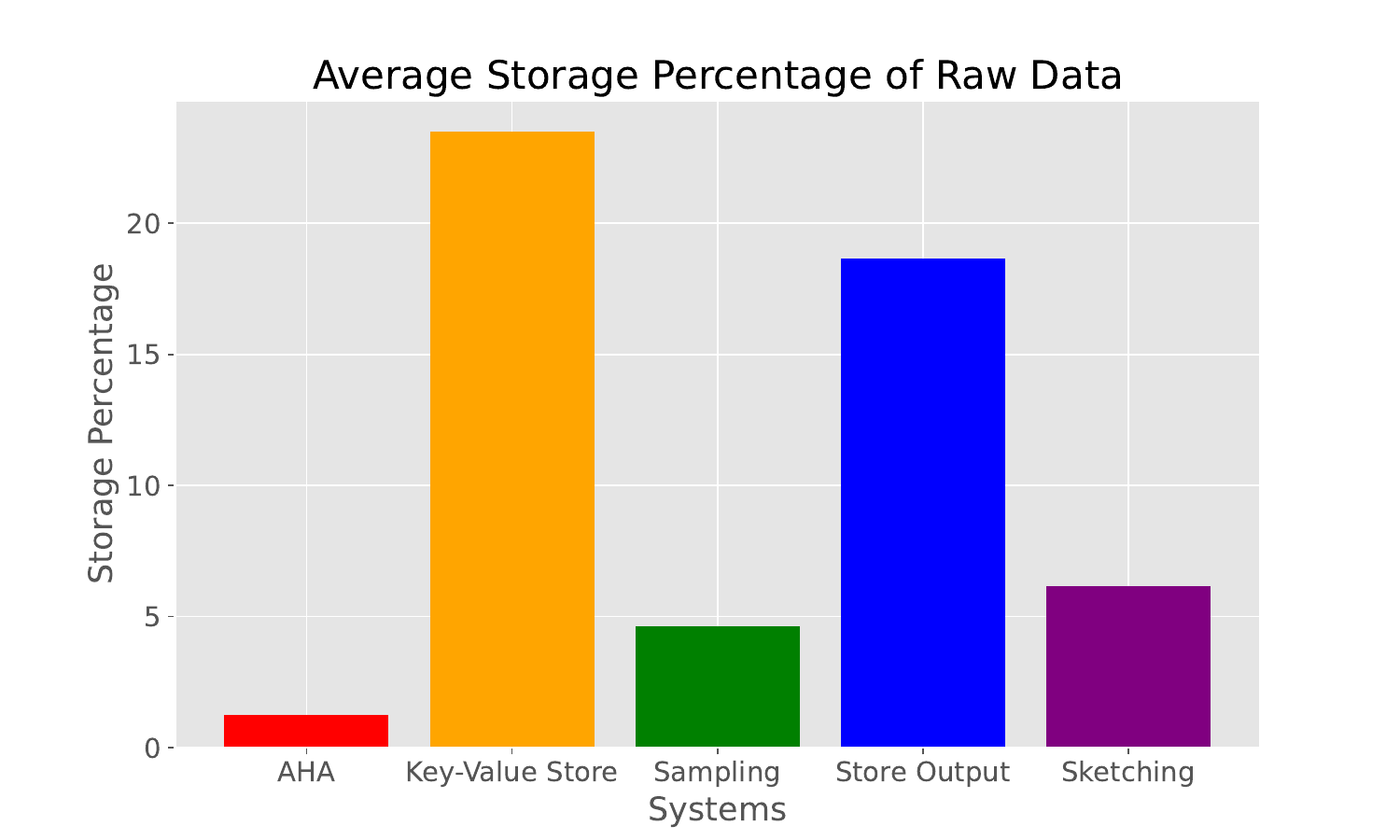}
  \caption{Percentage of storage of raw data. \model has the lowest storage requirements over other solutions.}
  \label{fig:storage_percent}
\end{figure}

\subsection{Results}
%We evaluate \model on cost, accuracy and scalability across a wide range of benchmark datasets.
\par \noindent \textbf{Accuracy and Cost trade-off:}
We show the cost vs. accuracy (fidelity) trade-off in Fig. \ref{fig:mainfig} and \ref{fig:subfigures_task} for \conviva and others in the Appendix Fig. \ref{fig:three_subfigures}.
First, we observe, as expected that \model, \rawingest, \kvstore and \cubeingest show 100\% fidelity.
\unifsample has an average accuracy of 71\% whereas \sketchsum has an
average accuracy of 88\%.
Further, the approximation methods have large variance in accuracy across groups and time.
The 90 percentile average accuracy of \unifsample and \sketchsum
are 27\% and 53\% respectively.
\model provides 34-85 times lower cost than the standard solutions of storing raw data or outputs features.
It also provides 6-10 times lower cost than \sketchsum.
%We also look at the cost distribution across storage and compute.
Since the compute cost per hour is costlier than the storage cost per GB per month, the storage cost is 3-4 orders of magnitude lower than the compute cost.
We measure the average reduction in storage in Fig. \ref{fig:storage_percent}.

\par \noindent \textbf{Scalability of \model:}
We study the scalability of \model and other solutions along two aspects.
First, we study the effect of an increase in the number of subgroups and user attributes.
This increases the number of leaf groups and the complexity
of the CUBE operation.
Approximation methods also include an increase in the number of groups with a smaller number of samples.
We use the synthetic data \synth where the attributes are generated
using a Zipf's distribution which mimics the average distribution of attributes in \conviva.
We let the number of user samples be the same as \conviva and arbitrarily increase the number of attributes to simulate the case where we have access to a larger number of user attributes, a common situation in video analytics.

\begin{figure}[h]
  \centering
  \begin{subfigure}[b]{0.48\linewidth}
    \centering
    \includegraphics[width=\textwidth]{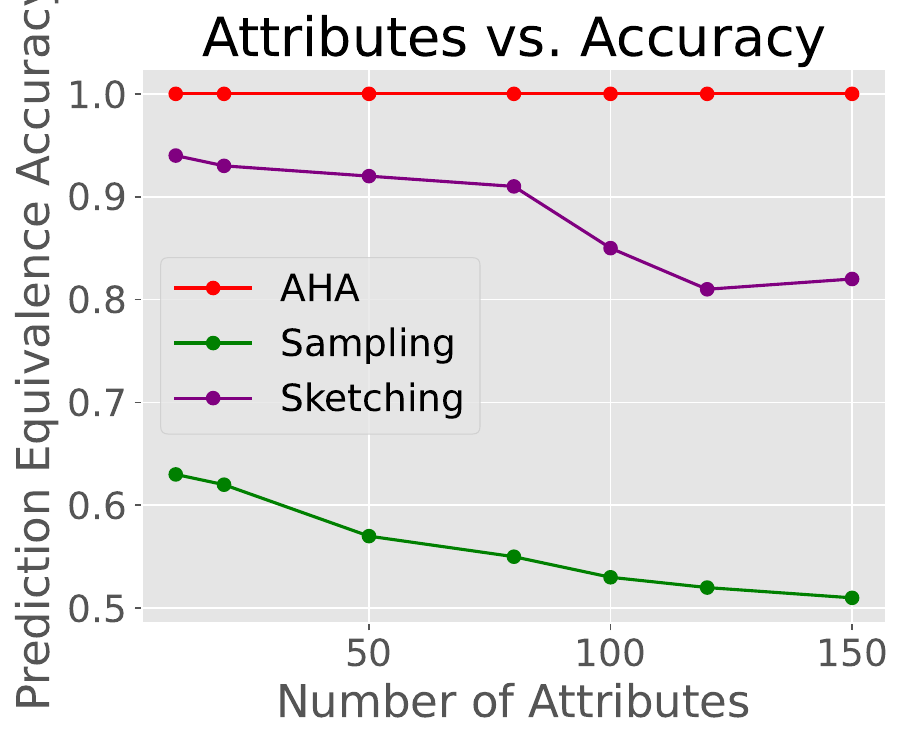}
    \caption{Scalability for accuracy}
    \label{fig:first_scale}
  \end{subfigure}
  %\hfill % Optional: it adds a space between the figures
  \begin{subfigure}[b]{0.48\linewidth}
    \centering
    \includegraphics[width=\textwidth]{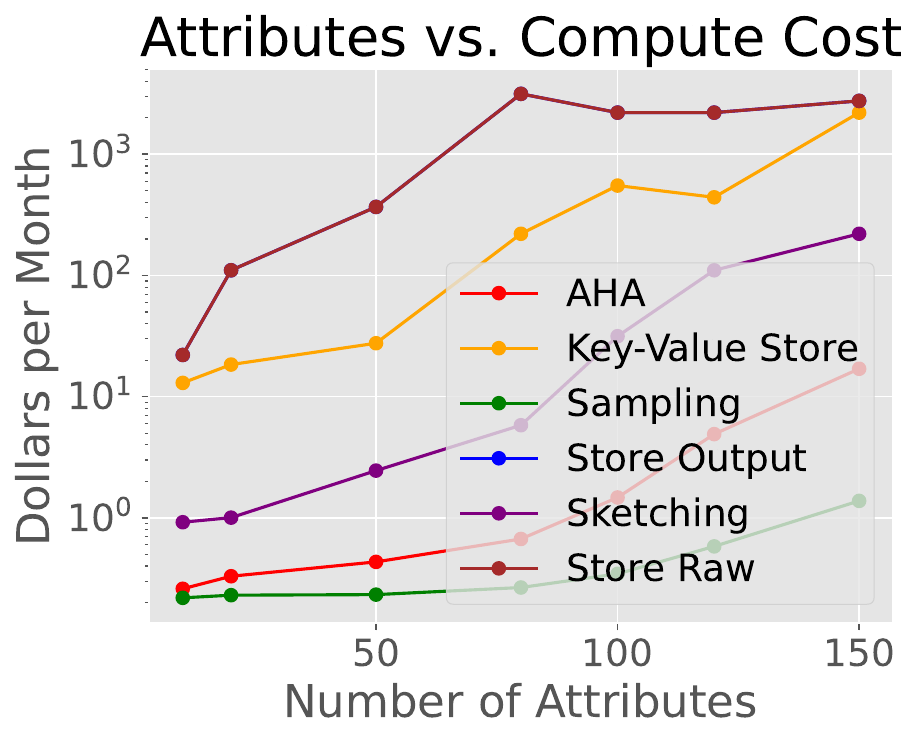}
    \caption{Scalability for cost}
    \label{fig:second_scale}
  \end{subfigure}
  \caption{\model provides cost scalability with increase in attributes while maintaining perfect accuracy.}
  \label{fig:scalability}
\end{figure}

The accuracy of approximate methods decreases with an increase in attributes (Fig. \ref{fig:first_scale}) due to
an increase in low sample groups. The compute cost
is two orders of magnitude lower than other perfectly accurate solutions as well as significant lower than sketching-based solution.
\model scales horizontally with no. of workloads.
We measure the scalability across multiple parallel workloads such as
when we have to perform analytics
for many customers simultaneously.
%We perform the parallel task on the same single node.
We measure the total cost as we increase the number of parallel workloads running each task on a node.
\model scales better with an increase in the number of workloads than other perfect replay solutions as well as the sketching solution (Fig. \ref{fig:parallel}).

\begin{figure}
  \centering
  \includegraphics[width=.54\linewidth]{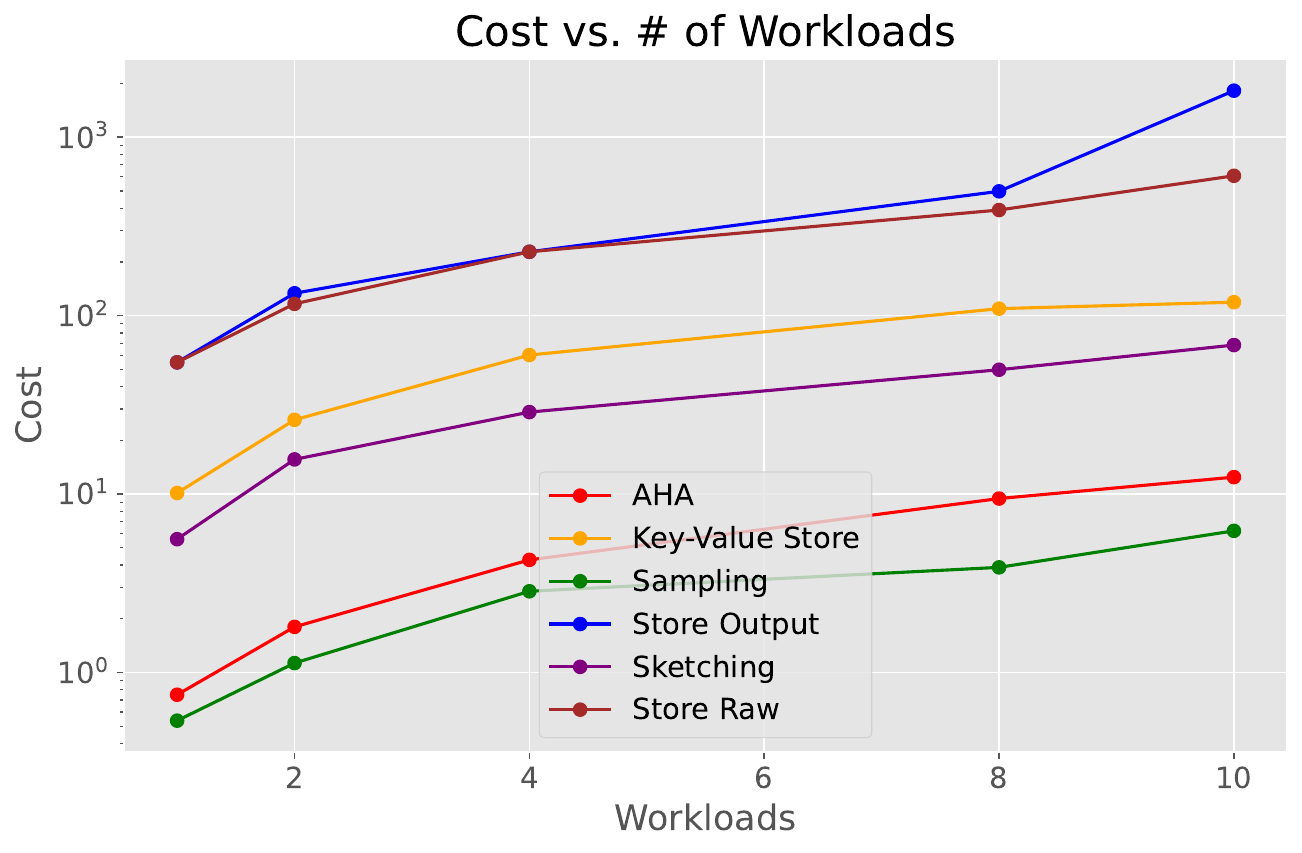}
  \caption{Cost over the number of workloads. \model scales efficiently with the number of workloads.}
  \label{fig:parallel}
\end{figure}

%\vyas{pls move this figure 11 before the conclusiosn :) }

\subsection{Deployment Experience}
We showcase the practical effectiveness of \model by providing a detailed
study of one of the core application in production.
The application involves user journey analytics on a mobile application used to stream video data. This study evaluates preprocessing and analytics performance tracking user navigation patterns.

\noindent
\emph{Setup:} While the total production dataset involves ingesting and processing hundreds of GB of data per minute, we look at a subset of dataset that focus on 10 user attributes measured over 30 days.
% Over the course of study the 26.17 GB of data contains 95,408,549 user sessions
We analyze a 30-day subset of production data focusing on 10 user attributes. This subset contains 26.17 GB of logs across 95,408,549 user sessions. We evaluate three key user engagement metrics: (i) successful seek operations, defined as playback resumption after a user-initiated seek, (ii) the number of recommended videos on the homepage that were played, and (iii) the number of carousel suggestions clicked.
% We analyze a 5.14GB dataset containing 15,138,822 user sessions collected over 30 days. Each row of the dataset includes 10 user attributes, with each entry representing a successful seek operation, where success is defined as using the seek feature followed by successful playback resumption. 
The quality of experience (QoE) metric counts these successful operations, grouped per minute during preprocessing. Since the metric of interest is the total number of operations per user group, we sum them across sessions—a decomposable operation that \model optimizes. Unlike the baseline’s repeated single-purpose aggregations using GROUP BY on raw data, \model uses a two-phase approach: creating a LEAF table with decomposable statistics (sum + count), followed by efficient CUBE-like roll-ups from pre-computed statistics rather than rescanning raw data.

\noindent
\emph{Performance comparison.} \model achieves a 1.25$\times$ speedup in per-minute data aggregation while producing identical output (45,454 cohorts). For the entire data, \model completed the operation in 2:30 hrs, while the baseline needed 3:08 hrs.
For downstream regression analysis using QoE metrics, \model demonstrates significant performance improvements on the processed data. Both methods achieve identical model quality (R² = 0.97), but \model's optimized data representation enables 6.2$\times$ faster query execution. This performance gain becomes crucial when scaling to such large scale datasets.
This example validates that \model maintains analytical accuracy while providing substantial performance improvements in both data preprocessing and query execution phases of the analytics pipeline.

\section{Conclusions}
Operational timeseries problems, characterized by high data volume, cardinality, and constant change, force practitioners to compromise cost and fidelity in retrospective analytics. This paper addresses these challenges by proposing a practical and rigorous solution: \model.
  \model offers a practical alternative where we can achieve {\em both} cost efficiency and perfect accuracy by  leveraging structural characteristics of the data,  query workloads, and modern analytical databases.    
 Our experiments and case-study shows \model provides 34-85 times less total cost of ownership without any loss in accuracy which translates to a cost savings of over \$0.7M per month. We observed about 75\% cost savings for important applications we deployed
    We believe our work is a first but significant step in  identifying and tackling a  practical operational problem of alternative history analytics. 

\textbf{Acknowledgements:} This work was supported in part by NSF (Expeditions CCF-1918770, CAREER IIS-2028586, Medium IIS-1955883, Medium IIS-2403240, Medium IIS-2106961, PIPP CCF-2200269), NIH (1R01HL184139), CDC MInD program, Meta faculty gifts, Dolby gift and funds/computing resources from Conviva and Georgia Tech.

\bibliographystyle{ACM-Reference-Format}
\bibliography{sample}
\appendix
\noindent\textbf{\Large Supplementary for the paper "\model:  Scalable  Alternative History Analysis for Operational Timeseries Applications"}
    % \section{Related Works}
% \label{sec:related}
% Previous works do not directly address the AHA problem we formulate in this paper. 
% We discuss some previous works on multidimensional data ingestion and retrieval 
% We can categorize then into approximate and precise methods. 

% \paragraph{Precise Methods:} Most standard operational pipelines
% either compute required features from data stream at ingestion time or later during fetch.
% Most of these methods use a database analytics solution like
% Spark, hadoop, Druid, etc~\cite{armbrust2015spark,dittrich2012efficient,yang2014druid}.
% These SQL and NoSQL databases usually compute the required features precisely and therefore provide 100\% accuracy.
% However, these solutions become expensive with increase in size of the data to store or with exponential increase in user groups as number of attributes increase as shown in our real-world applications~\cite{agarwal2013blinkdb}.

% \paragraph{Approximate Methods:} These methods compute approximation of a metric usually at lower compute and memory cost. 
% Online Sampling is a popular technique that reduces compute and latency costs~\cite{chaiken2008scope,ting2019approximate}. Sketch based methods provide a more efficient bounded tradeoff between memory and accuracy with sublinear complexity.
% These frameworks generate approximate summaries during ingest and use
% them to estimate statistics with apriori provable error bounds~\cite{manousis2022enabling,cormode2011synopses,durand2003loglog}.

\section{Applicability of \model to different algorithms}

One trivial sample statistic that can be used to derive any sufficient statistic is the set of
user data samples themselves (i.e. we store simply store the full raw dataset).
However, due to storage cost, we cannot support such algorithms~\cite{lohrer2023gadformer,xiong2011group}
that need to ingest all the samples.
We, therefore, review commonly used sample statistics and characterize their decomposability and applicability
to \model framework.

We can calculate the mean of any parent \subgroup by using the count and sum of children \subgroups
which are self-decomposable.
Therefore, we can support the mean for any \subgroup.
We cannot accurately calculate the median of the \subgroup population without all the samples.
However, there are reliable approximation algorithms that involve storing frequency values of histograms
of the samples~\cite{shrivastava2004medians}. %and guarantee $O(\frac{\log \sigma}{n})$ approximation error
%where $\sigma$ is the sample standard deviation and $n$ is the size of the samples.
A similar method can also be used to approximate estimates of sample quantiles and percentiles.
Variance and standard deviation of a sample can be estimated from the mean, count, and sum of squares
all of which are decomposable. Similarly, we can calculate higher-order moments by tracking the sum of
sample values raised to appropriate exponents.
%This enables us to derive statistics that capture the shape of distribution like skewness and kurtosis
%which depend on third and fourth-order moments.
Similar to median and quartiles, we cannot accurately estimate the interquartile range but can leverage
approximate estimates of quartiles.
Since maximum and minimum are self-decomposable statistics, we can accurately estimate the range.
%Computing median absolute error can be approximated using histogram-based binning of samples to arbitrary
%accuracy~\cite{chen2021approximating}.\harsha{Entropy, Gini coefficient?}
%\paragraph{Other Statistics}
%Among order statistics, we can trivially observe that range, minimum and maximum are supported
%since they are self-decomposable.
%While any arbitrary order statistic cannot be supported, there are well-known approximation algorithms
%for cases where we assume an underlying parametric distribution such as Gaussian~\cite{royston1982expected}.
%We can also support widely used test statistics such as t-statistics and z-tests which typically rely only
%on mean and moments of the sample distribution.
%We can also support correlation measures between multiple metrics such as Pearson correlation coefficient
%which depends on the sum of each of the metrics as well as the sum of the product of the metrics
%which are self-decomposable.
%\harsha{More correlation statistics, longitudinal statistics}
%\subsubsection{Applicability to widely used anomaly detection algorithms}
%First, we note that the LAD problem has not been studied and hence we do not have any proven anomaly detection algorithms tailored to it.
Most standard benchmarks for standard temporal anomaly detection do not make any assumptions on the type of features used for generating the time-series or use simple statistics like mean or sum of observations from a group of sensors~\cite{darban2022deep}.
%For example, \cite{su2019robust, hundman2018detecting, mathur2016swat} use single source time-series such as telemetry data from individual servers or spacecraft sensors.

%\vyas{we should maybe talk about "design" vs "implementation" -- eg the leaf/cube decoupling can be implemented on many platforms including spark and clickhouse. which is a strength of our approach? }

\section{Benchmark Setup Details}
\label{sec:benchmark}
We evaluate all the baselines on a standard compute setup consisting of a 64-core Intel Xeon CPU and 256 GB of available main memory.
The input data is stored as a parquet file on the single node for each time step.
The summarized data is stored as a compressed CSV file for \model (we use the fast zstd compression). To fetch the relevant statistics for prediction the CSV file is extracted, and then we roll up to evaluate the statistics for the required groups.

\subsubsection{Anomaly detection algorithms:}
To analyze the ability of \model and baselines to support a wide range of prediction pipelines
we consider different anomaly detection algorithms.
We use three diverse anomaly detection algorithms that are being used widely in research and practice~\cite{blazquez2021review,chandola2009anomaly,han2022adbench}.  We first evaluate with a simple baseline that detects if any metric's average value is 3 standard deviations away from the historical mean (\meansigma)~\cite{su2019robust}.
Finally, we use two other popular anomaly detection algorithms: $k$ nearest neighbors (\knn)~\cite{amer2012nearest}
and Isolation Forest (\isof)~\cite{liu2008isolation}.
All these algorithms use mean of the metrics observed across the users as features to detect the anomalies.
We do not choose complex neural algorithms since the time required to compute the anomalies far exceeds the time for the replay system to ingest, store and fetch the data. However, they also use similar statistics used by these algorithms.
We show the results for the average accuracy of three tasks here since they are observed to be very similar to each other.
Moreover, over 99\% of the total cost is due to ingesting and fetching group statistics for all three tasks.

\subsubsection{Evaluation metrics:}
We measure the fidelity of each of the baselines by measuring both the accuracy of the aggregate metric (metric accuracy) as well as the accuracy of the anomaly detection algorithm (task accuracy).
We measure metric accuracy via the RMSE score and task accuracy via the classification accuracy of anomalies.
Since \model has perfect recall of the required statistics for decomposable statistics,
we expect the agreement accuracy to be 100\% for both the metric and the algorithm predictions, as are other perfect recall solutions: \rawingest, \cubeingest, \kvstore.
Due to approximations of \unifsample and \sketchsum, there may be
reduction in metric and task accuracy of these solutions.
We measure the average accuracy across time and across all user subgroups.
We also measure the variance in accuracy across subgroups since subgroups with small samples are particularly vulnerable to mispredictions by these methods.

\subsubsection{Cost Model:} To determine the total cost of the system (compute and storage), we measure the cost of performing this operation over a period of a month
on a EC2 system of similar configuration via AWS using S3 to store the
summarized data for each solution.
We measured the total compute time and multiplied it with cost of compute  (\$0.96 per hour) and similarly calculated the total storage cost by calculating the cost per GB (\$0.15 per month).
The combined cost is measured. We report the normalized cost which is the total cost divided by the cost of the default typical big data architecture solution \rawingest e.g. Spark.

\subsubsection{Calibrating hyperparameters for weak equivalence models:}
Weak equivalence methods have hyperparameters that determine the tradeoff between accuracy and cost
of the solution.
\unifsample has a hyperparameter $k$ which is the percentage of samples to be stored.
Similarly, \sketchsum has a hyperparameter $s$ which is the size of the sketch.
The size of the sketch is the number of hash functions used in the sketch times the number of buckets in the sketch.
We tune the number of hash functions of \sketchsum.
We tune the hyperparameters based on user requirements defined by parameters ($\delta, \epsilon$)
where $\delta$ is the percentage of groups that can have an error of at most $\epsilon$.
By default we set $\delta=0.95$ and tune for different $\epsilon$ to determine the cost trade-off.
\begin{table}[h]
  \scalebox{0.65}{
    \begin{tabular}{c|c|c|c|c | c |c}
      Dataset         & \conviva  & \trace & \taxi & \chbench & \mgbench & \bigdata \\\hline
      Attributes      & 12       & 3      & 2     & 32       & 4        & 31           \\
      Metrics         & 7         & 1      & 8     & 11       & 15       & 7             \\
      Avg. \# groups  & 1303562   & 344    & 513   & 28711    & 722      & 119552      \\
      Timesteps       & 1081        & 2351   & 730   & 7116     & 271      & 1998        \\
      Total size (GB) & 1667        & 48     & 2463  & 17.2     & 31.5     & 117.8       \\
    \end{tabular}
  }
  \caption{Datasets statistics}
  \label{tab: datasets}
\end{table}

\begin{figure*}[h]
  \centering
  \begin{subfigure}[b]{0.325\textwidth}
    \centering
    \includegraphics[width=\textwidth]{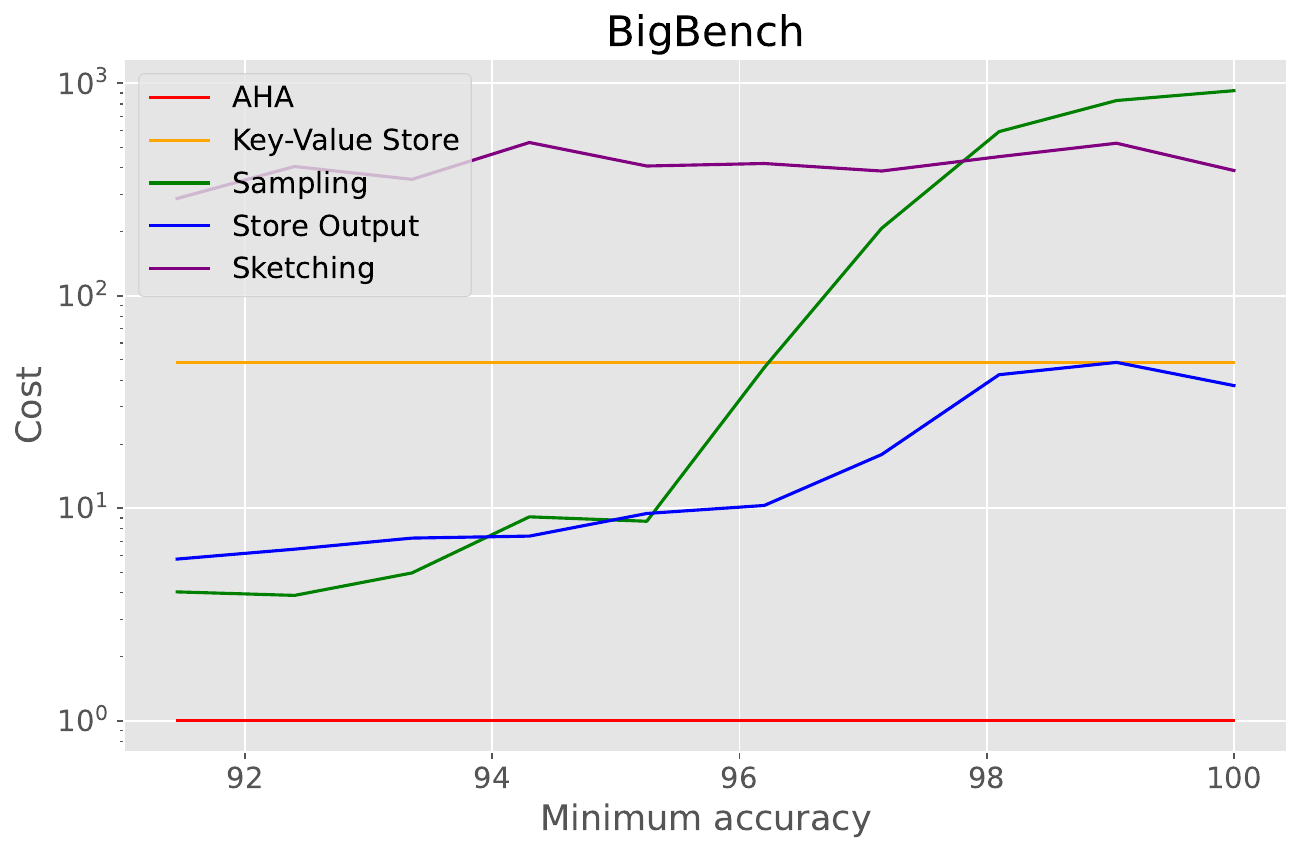}
    \caption{\bigdata}
    \label{fig:first}
  \end{subfigure}
  \hfill % Optional: it adds a space between the figures
  \begin{subfigure}[b]{0.325\textwidth}
    \centering
    \includegraphics[width=\textwidth]{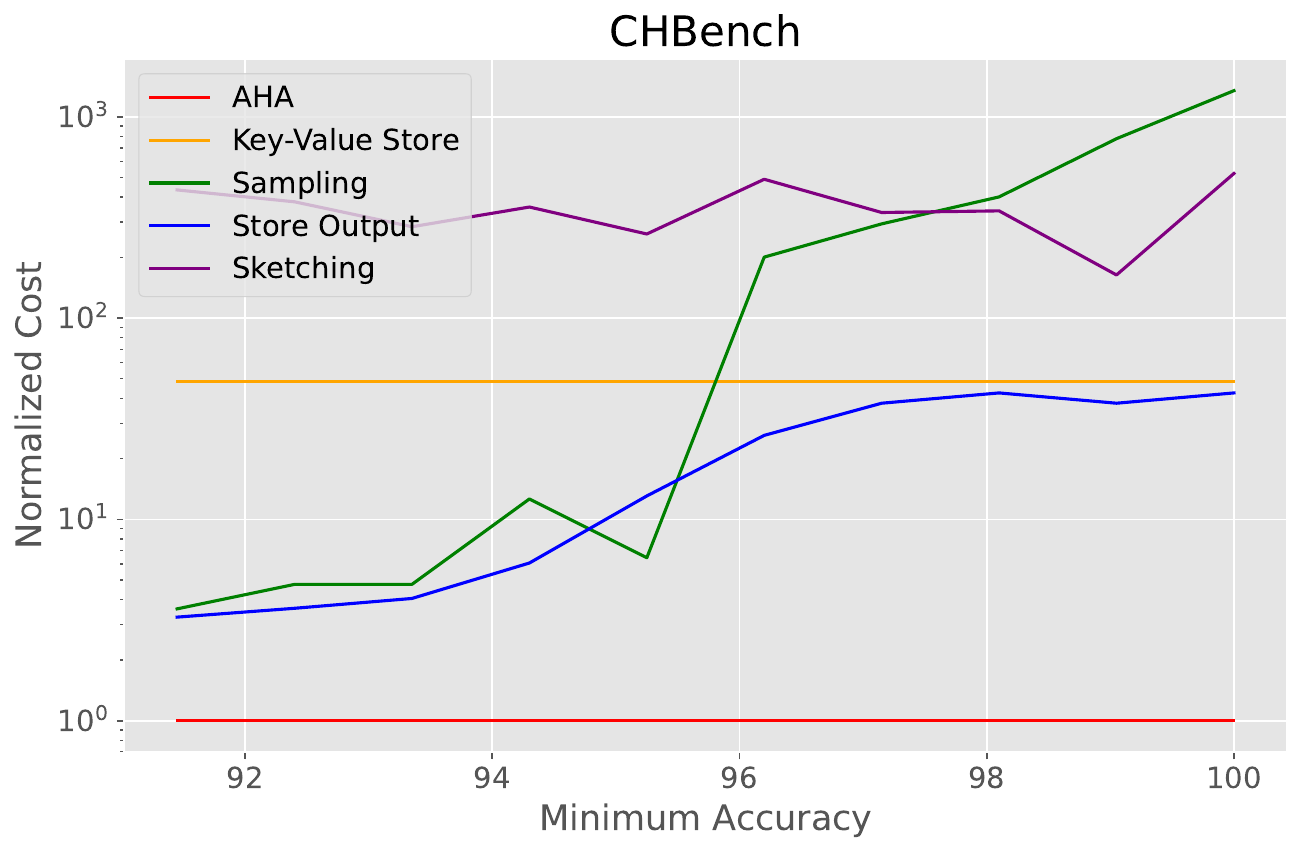}
    \caption{\chbench}
    \label{fig:second}
  \end{subfigure}
  \begin{subfigure}[b]{0.325\textwidth}
    \centering
    \includegraphics[width=\textwidth]{RoughFigs/q3/videoanalytics.pdf}
    \caption{\conviva}
    \label{fig:third}
  \end{subfigure}
  \hfill
  \begin{subfigure}[b]{0.325\textwidth}
    \centering
    \includegraphics[width=\textwidth]{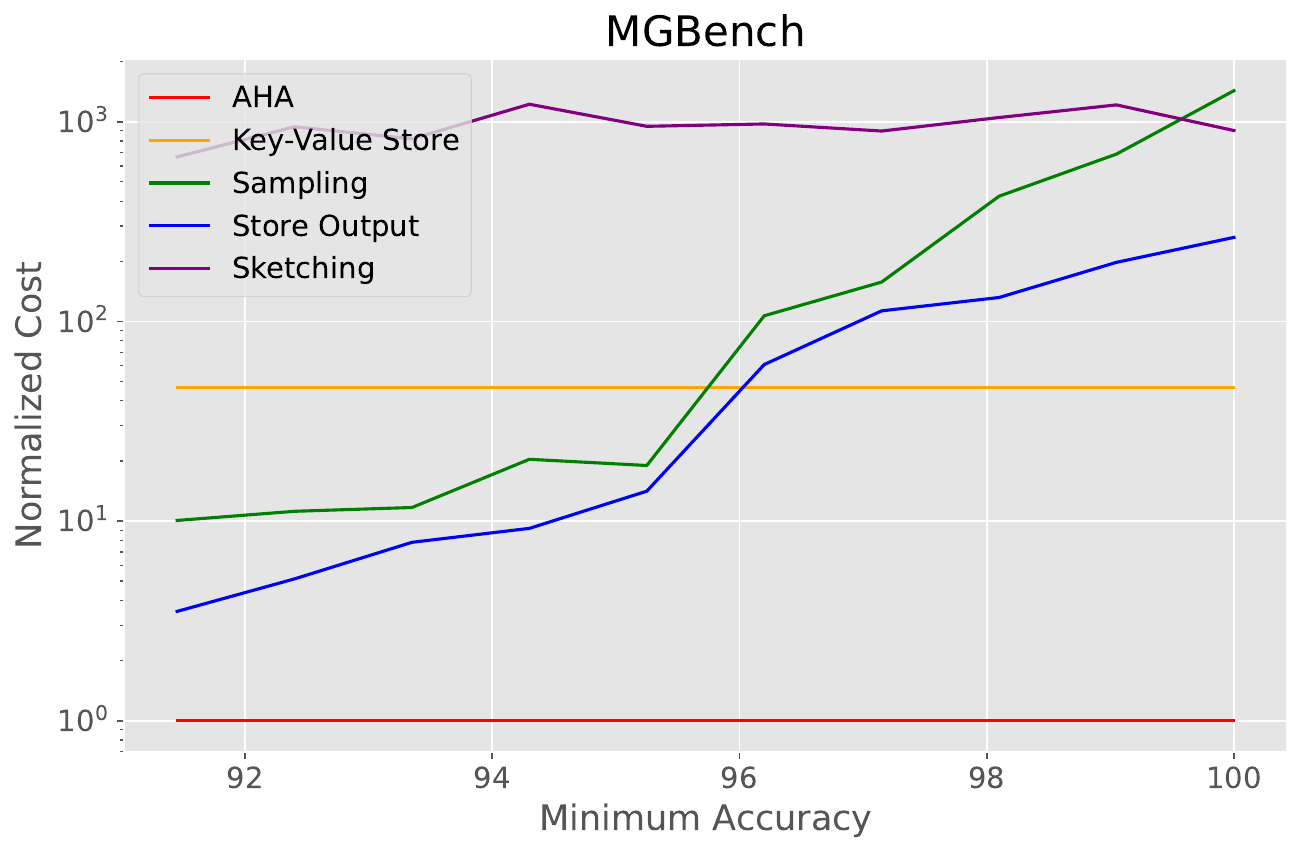}
    \caption{\mgbench}
    \label{fig:fourth}
  \end{subfigure}
  \begin{subfigure}[b]{0.325\textwidth}
    \centering
    \includegraphics[width=\textwidth]{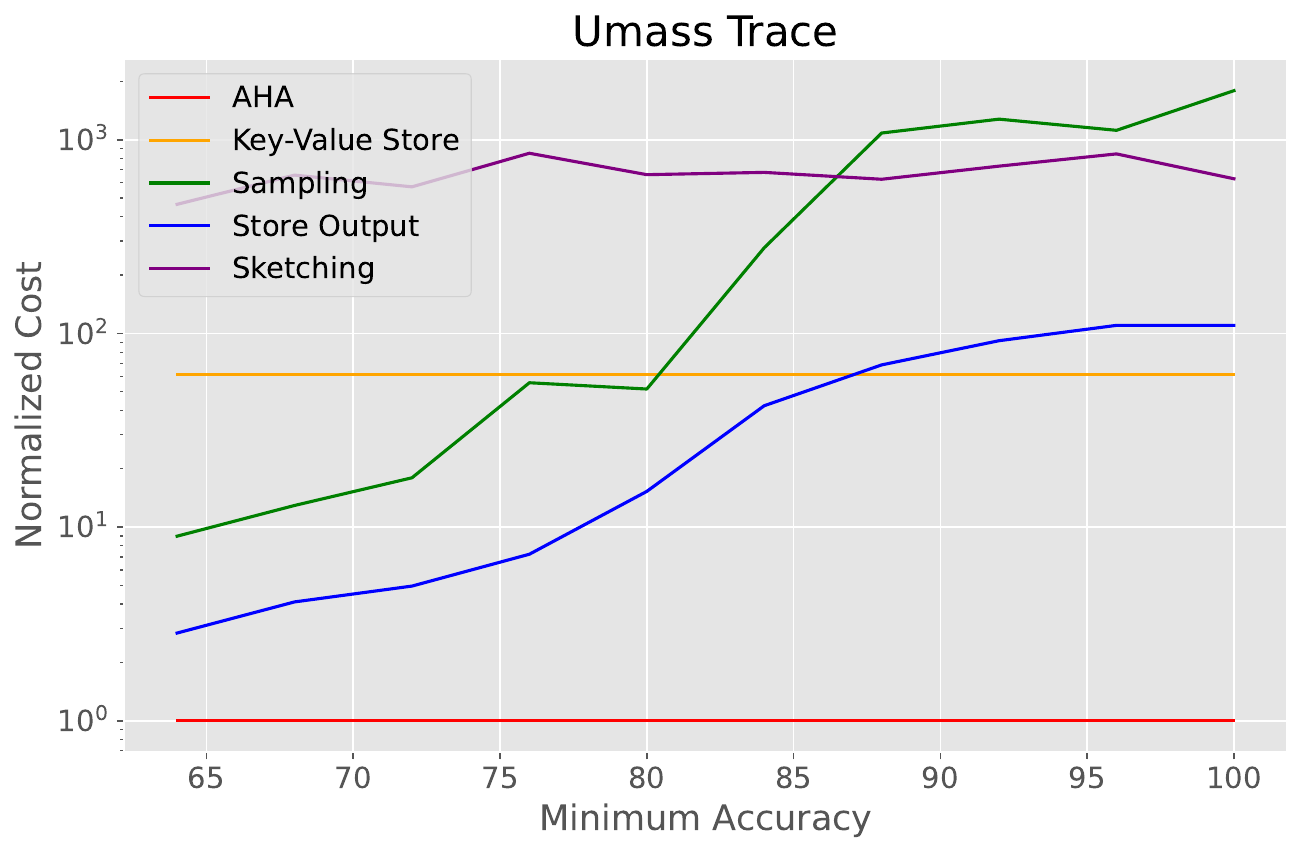}
    \caption{\trace}
    \label{fig:fifth}
  \end{subfigure}
  \hfill % Optional: it adds a space between the figures
  \begin{subfigure}[b]{0.325\textwidth}
    \centering
    \includegraphics[width=\textwidth]{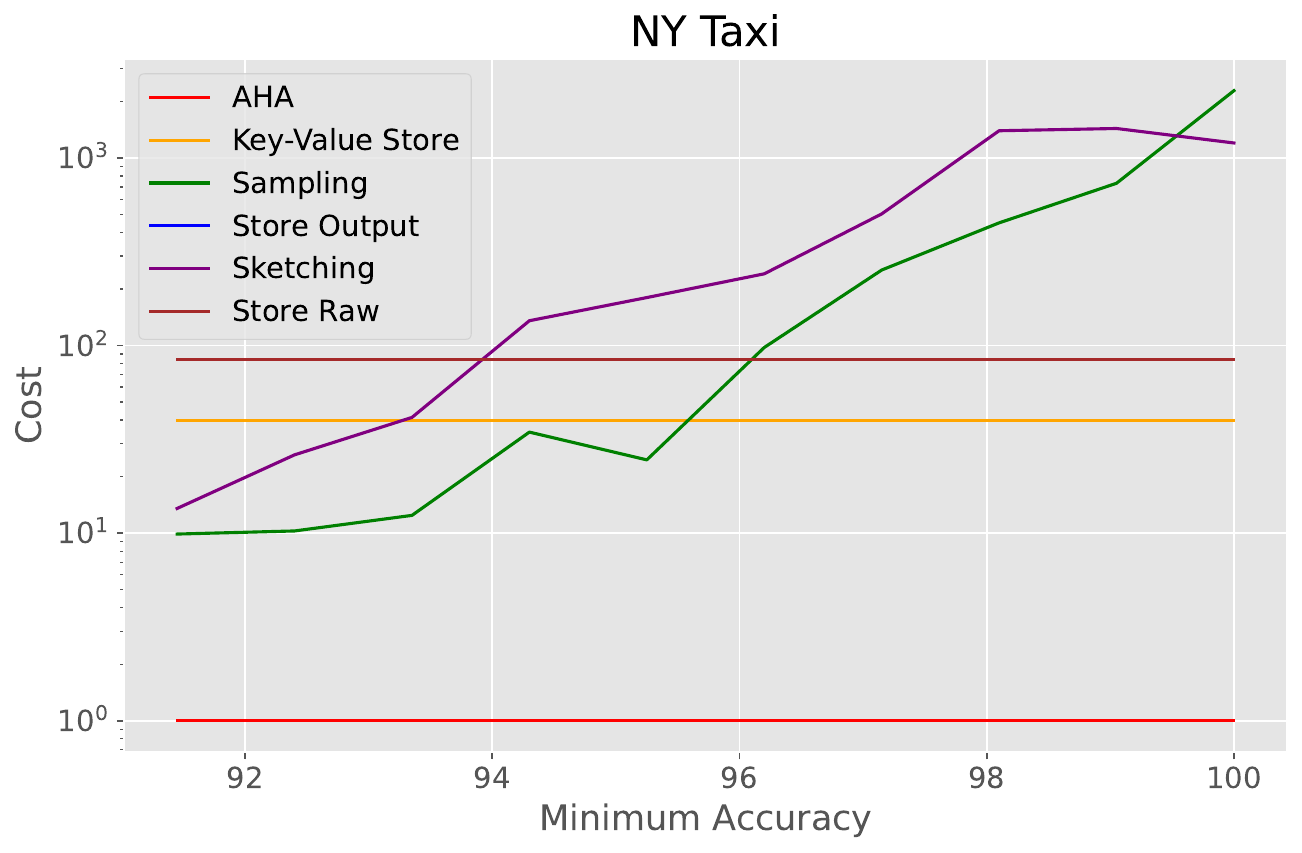}
    \caption{\taxi}
    \label{fig:sixth}
  \end{subfigure}
  \caption{Normalized Cost for different datasets for different minimum task accuracy requirements. 
  \model is the only solution with the lowest cost across all accuracy requirements (> 90\%) with 8.6-87 times lower cost. Also note that for accuracy > 97\% weak-equivalence methods cost more than even storing full session data showcasing compute overhead of sampling and sketching methods for high accuracy prediction.}
  \label{fig:three_subfigures_task}
\end{figure*}

\begin{figure*}[hbp]
  \centering
  \begin{subfigure}[b]{0.325\textwidth}
    \centering
    \includegraphics[width=\textwidth]{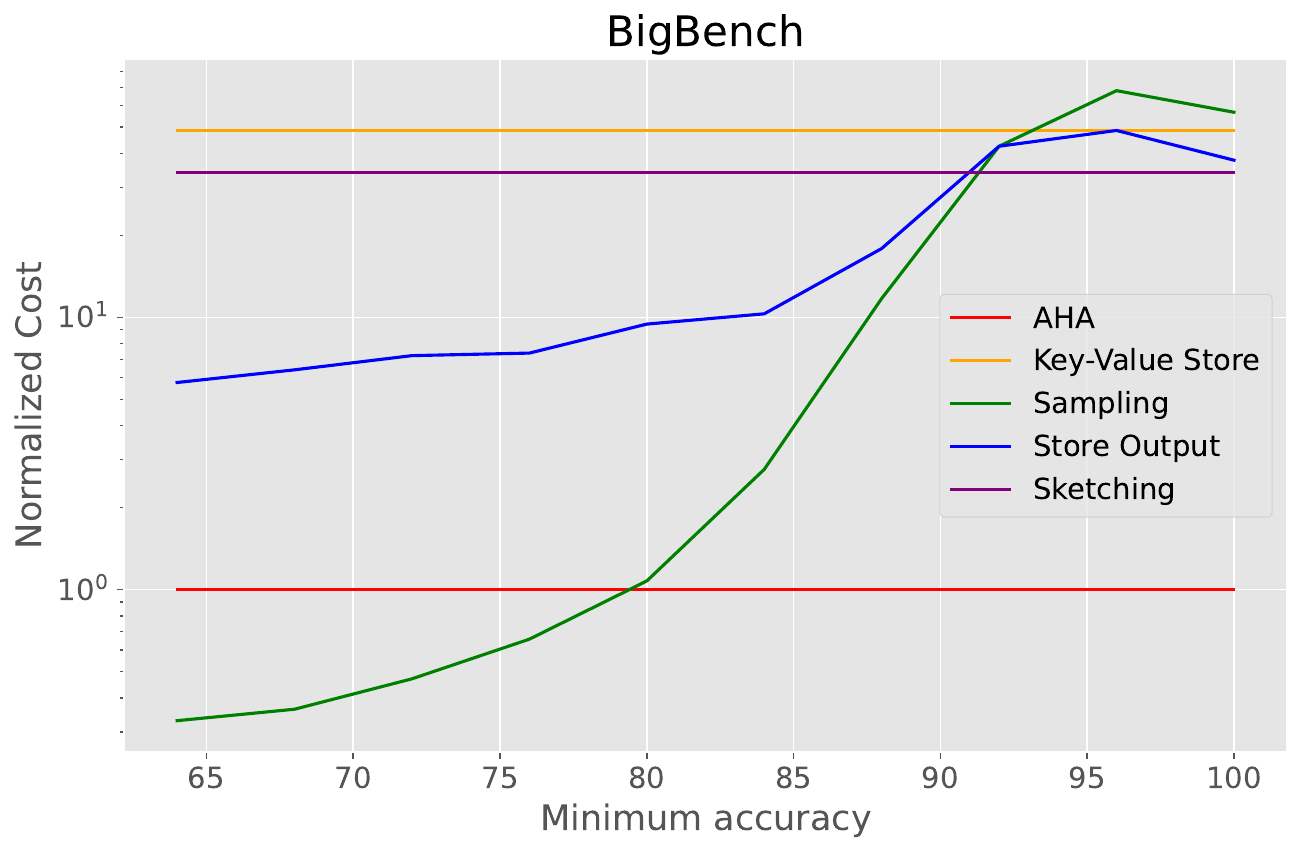}
    \caption{\bigdata}
    \label{fig:first}
  \end{subfigure}
  \hfill % Optional: it adds a space between the figures
  \begin{subfigure}[b]{0.325\textwidth}
    \centering
    \includegraphics[width=\textwidth]{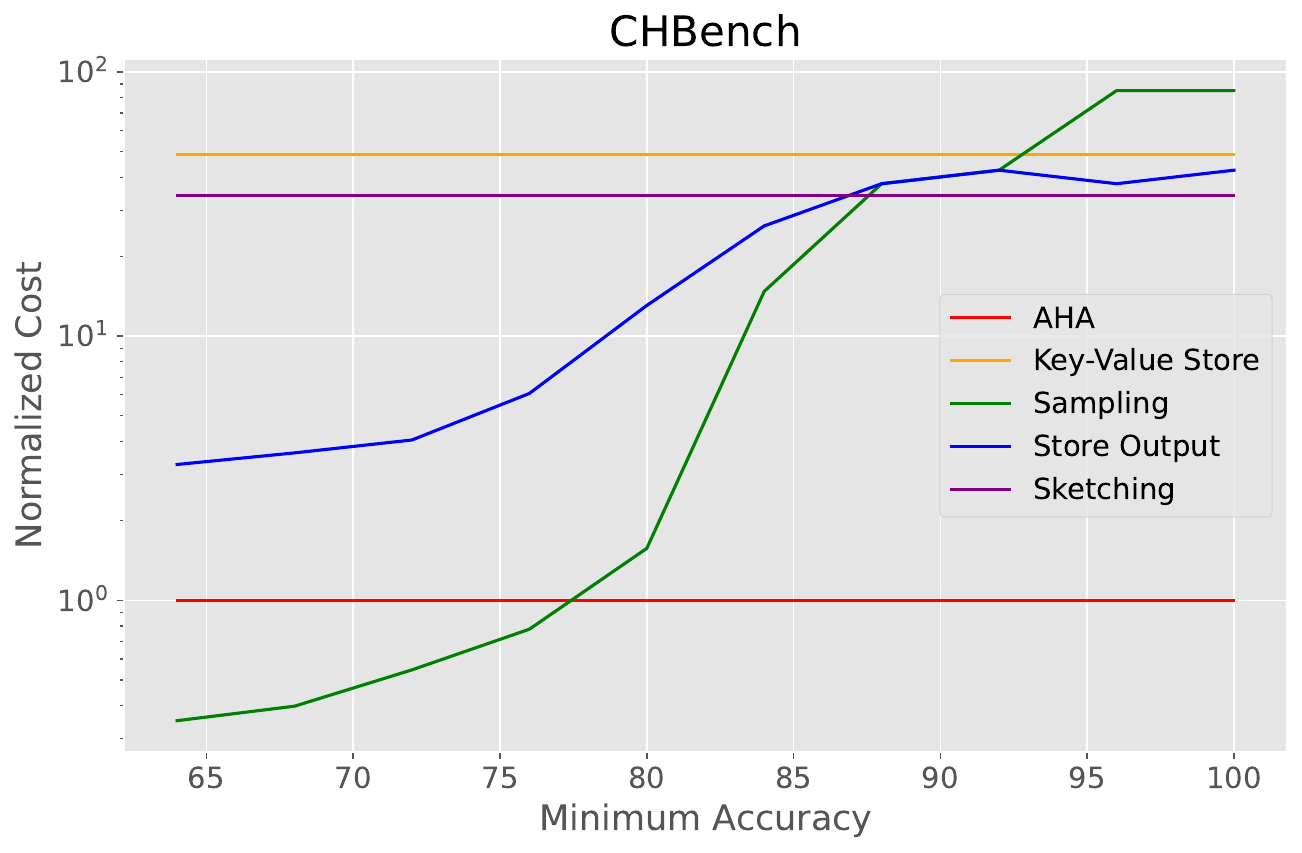}
    \caption{\chbench}
    \label{fig:second}
  \end{subfigure}
  \begin{subfigure}[b]{0.325\textwidth}
    \centering
    \includegraphics[width=\textwidth]{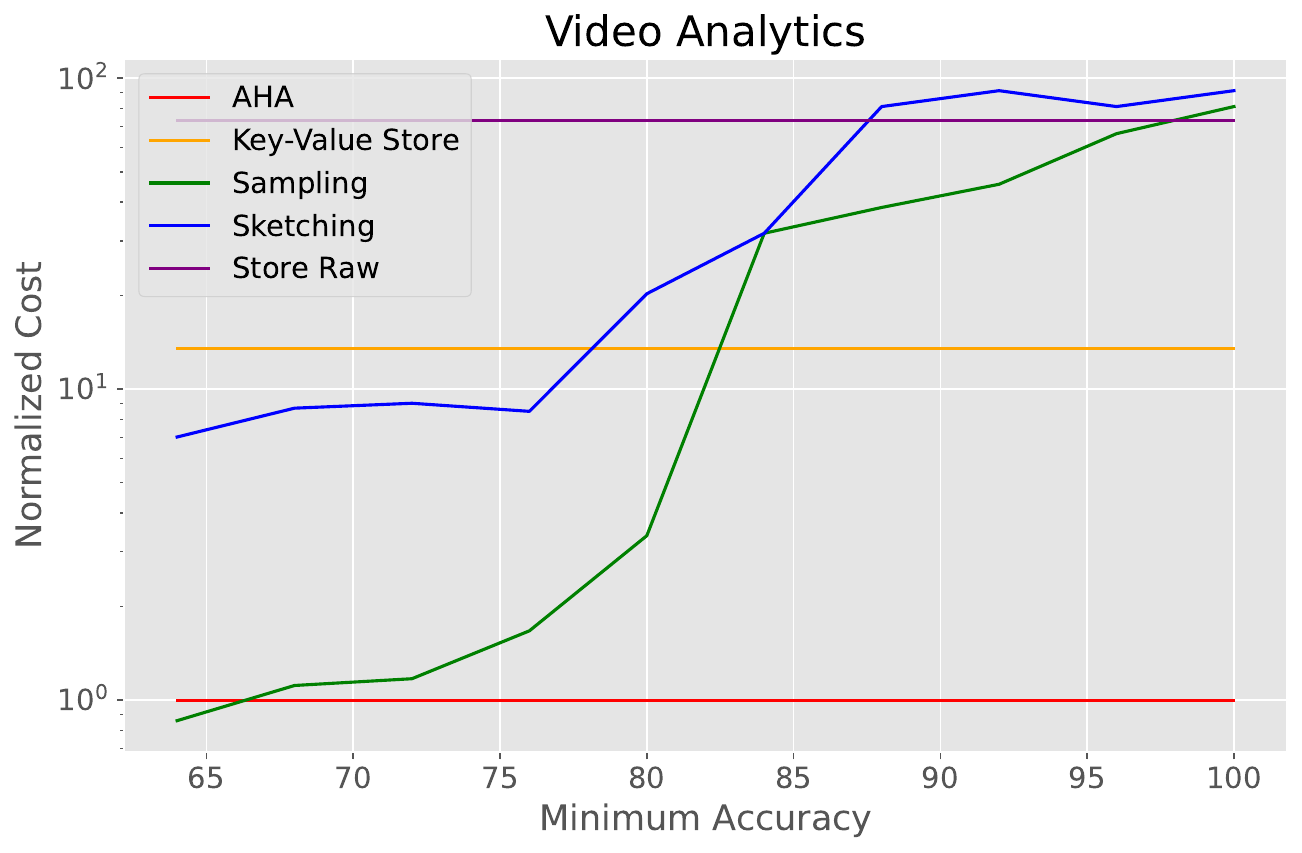}
    \caption{\conviva}
    \label{fig:third}
  \end{subfigure}
  \hfill
  \begin{subfigure}[b]{0.325\textwidth}
    \centering
    \includegraphics[width=\textwidth]{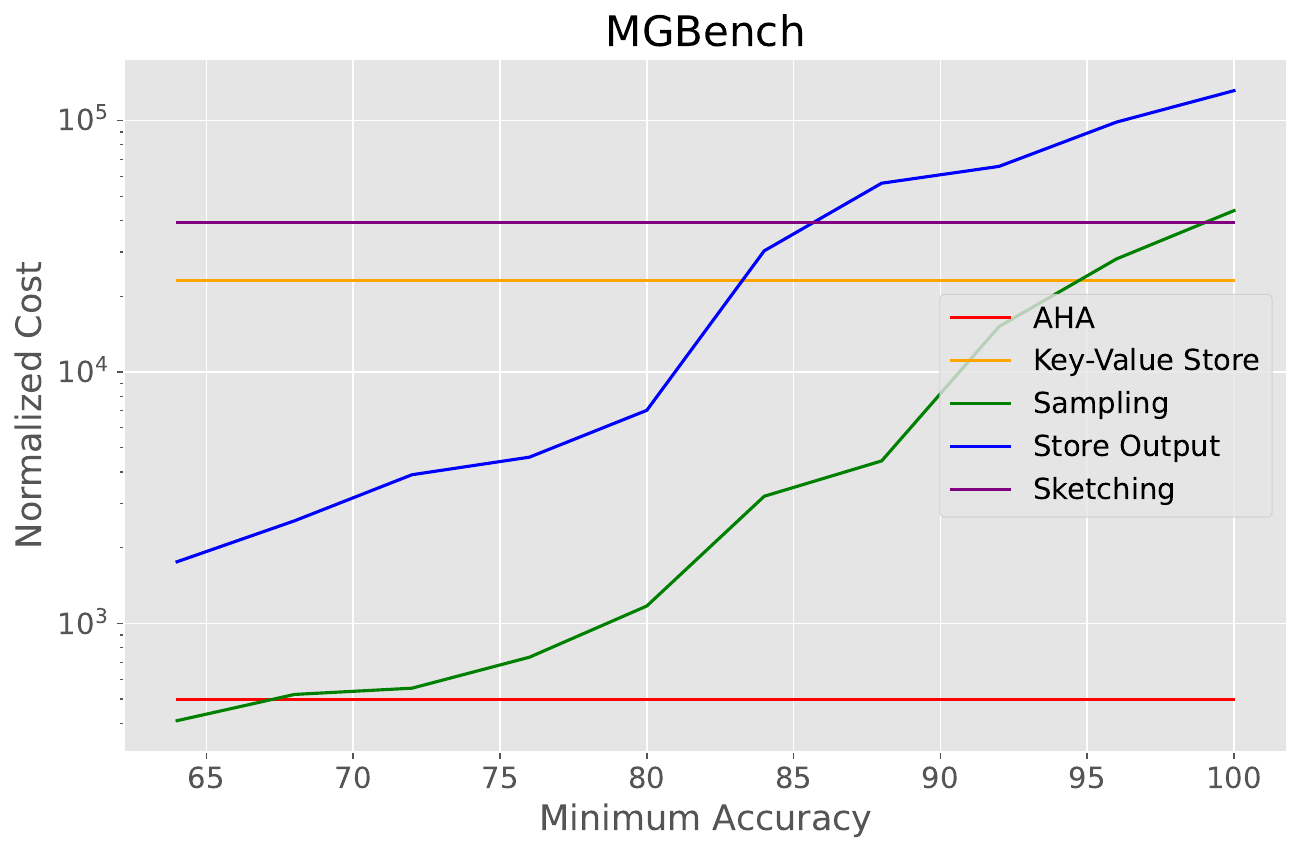}
    \caption{\mgbench}
    \label{fig:fourth}
  \end{subfigure}
  \begin{subfigure}[b]{0.325\textwidth}
    \centering
    \includegraphics[width=\textwidth]{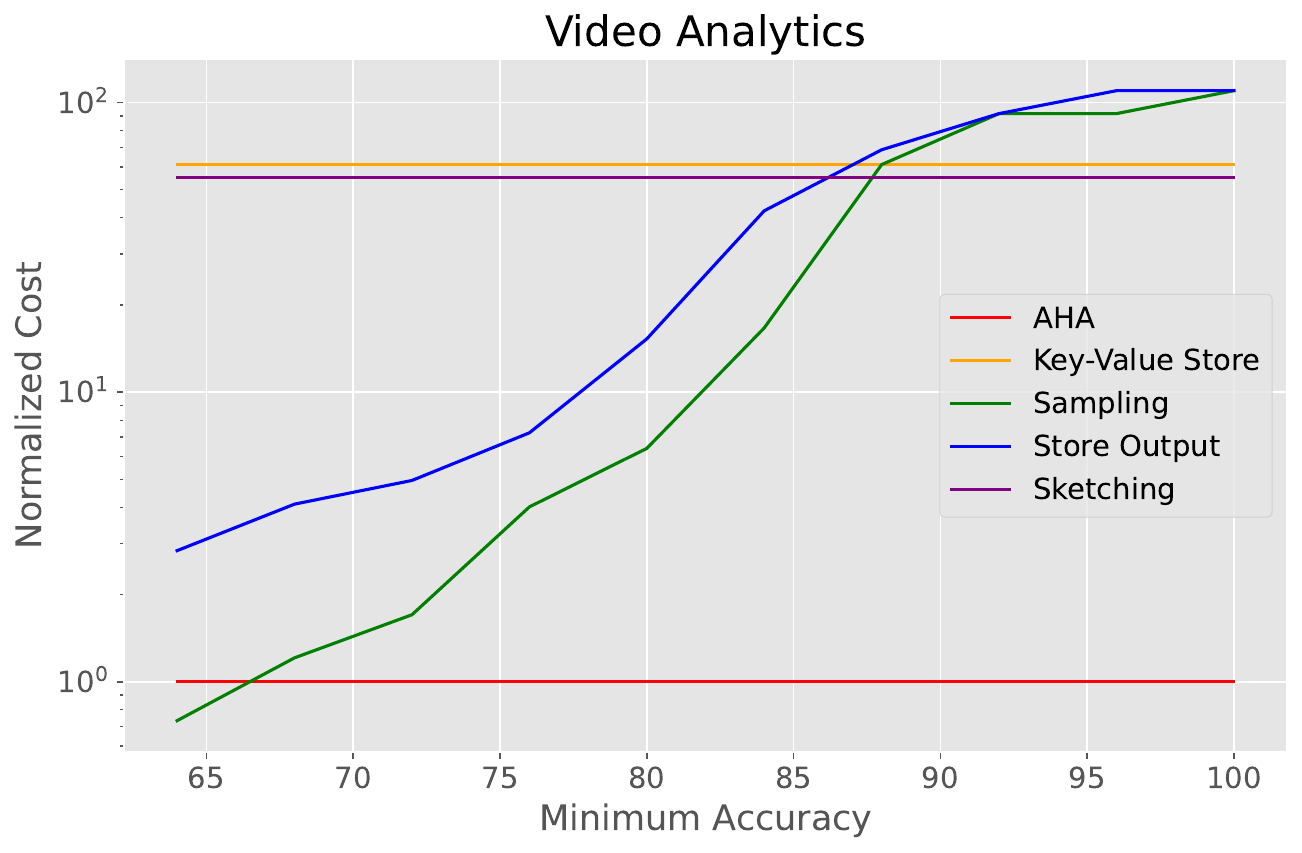}
    \caption{\trace}
    \label{fig:fifth}
  \end{subfigure}
  \hfill % Optional: it adds a space between the figures
  \begin{subfigure}[b]{0.325\textwidth}
    \centering
    \includegraphics[width=\textwidth]{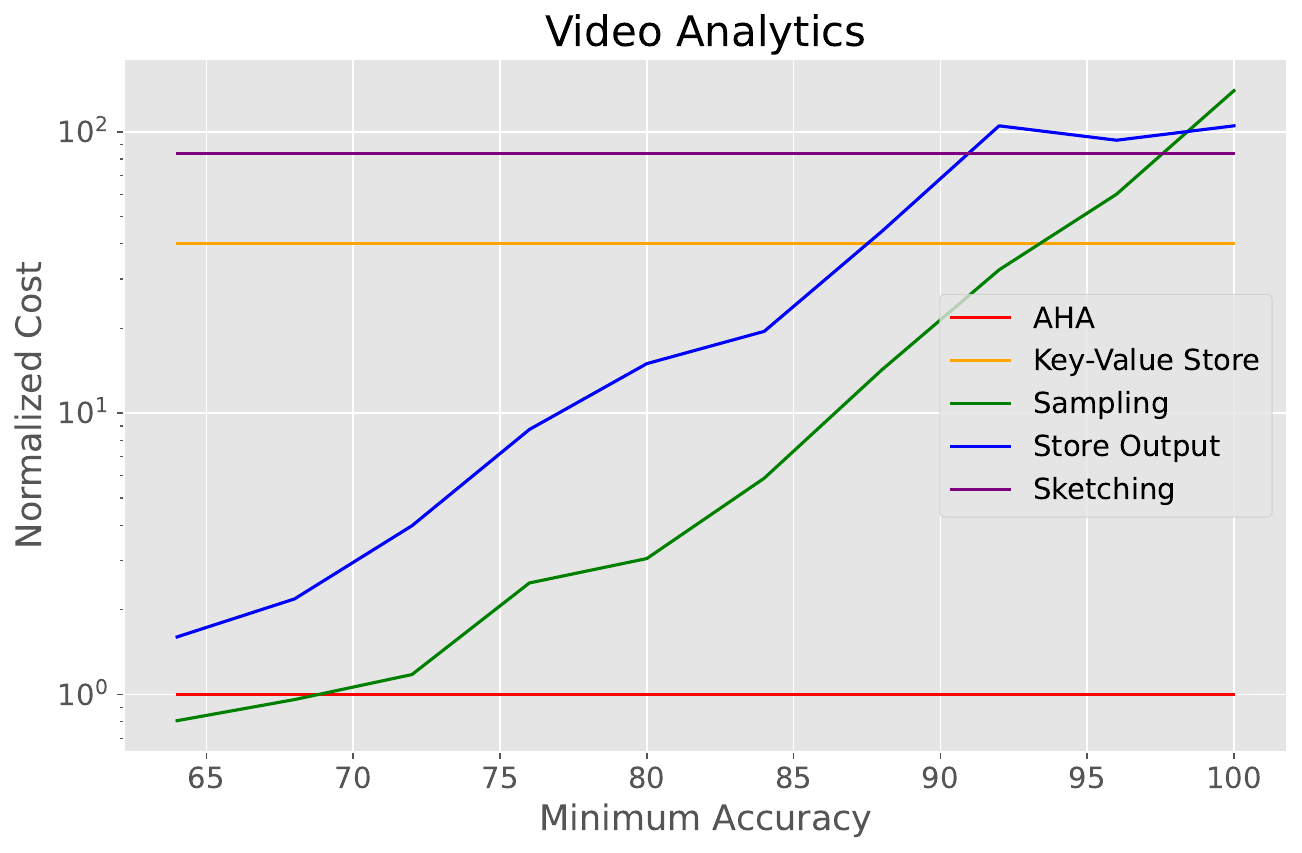}
    \caption{\taxi}
    \label{fig:sixth}
  \end{subfigure}
  \caption{Normalized Cost for different datasets for different minimum metric accuracy requirements. 
  \model is the only solution with the lowest cost. \model is 55-130 times more efficient among strong equivalence solutions and are the cheapest solutions for all datasets with over 80\% accuracy requirement.  }
  \label{fig:three_subfigures}
\end{figure*}

\end{document}